\documentclass[preprint,12pt]{elsarticle}
\pdfoutput=1

\usepackage{amssymb}
\usepackage{amsthm}
\usepackage{multirow}
\newtheorem{theorem}{Theorem}[section]
\usepackage{graphicx,color}
\usepackage{amsmath,epsfig}
\usepackage{times,amsmath,epsfig}
\usepackage[utf8x]{inputenc}

\usepackage{amsmath, amsfonts}
\usepackage{amssymb,epsf,alltt}
\usepackage{epsfig}

\usepackage[ruled,vlined]{algorithm2e}
\usepackage{algorithmic}
\usepackage{url}

\usepackage{setspace}

\newcommand{\REM}[1]{}
\newtheorem{lemma}{Lemma}

\numberwithin{equation}{section}

\journal{Information Systems, Elsevier}

\begin{document}

\begin{frontmatter}



\title{Advanced Bloom Filter Based Algorithms for Efficient Approximate Data De-Duplication in Streams}

\author{Suman K. Bera}
\ead{sumanber@in.ibm.com}
\address{IBM Research, India}

\author{Sourav Dutta}
\ead{sodutta7@in.ibm.com}
\address{IBM Research, India}

\author{Ankur Narang}
\ead{annarang@in.ibm.com}
\address{IBM Research, India}

\author{Souvik Bhattacherjee\corref{cor1}}
\ead{bsouvik@cs.umd.edu}
\address{University of Maryland, College Park, USA}
\cortext[cor1]{\textbf{This work was completed at IBM Research, India.}}




\begin{abstract}
	Data intensive applications and computing has emerged as a central area of modern research with the explosion of 
	data stored world-wide. Applications involving telecommunication call data records, web pages, online transactions, 
	medical records, stock markets, climate warning systems, etc., necessitate efficient management and processing of 
	such massively exponential amount of data from diverse sources. Duplicate detection and removal of redundancy 
	from such multi-billion datasets helps in resource and compute efficiency for downstream processing. De-duplication 
	or \emph{Intelligent Compression} in streaming scenarios for approximate identification and elimination of duplicates 
	from such unbounded data stream is a greater challenge given the real-time nature of data arrival. Stable Bloom Filters 
	(SBF) addresses this problem to a certain extent. However, SBF suffers from a high false negative rate and slow 
	convergence rate, thereby rendering it inefficient for applications with low false negative rate tolerances.

	In this work, we present several novel algorithms for the problem of approximate detection of duplicates in data streams. 
	We propose the \emph{Reservoir Sampling based Bloom Filter} (RSBF) combining the working principle of reservoir sampling 
	and Bloom Filters. We also present variants of the novel \emph{Biased Sampling based Bloom Filter} (BSBF) based on 
	biased sampling concepts. Using different updation and biasing mechanisms we propose variants of the same model enabling 
	the data structure to adapt to various input scenarios. We also propose a randomized load balanced variant of the sampling 
	Bloom Filter approach to efficiently tackle the duplicate detection. In this work, we thus provide a generic framework for 
	de-duplication using Bloom Filters. Using detailed theoretical analysis we prove analytical bounds on the false positive rate, 
	false negative rate and convergence rate of the proposed structures. We exhibit that our models clearly outperform the existing 
	methods. We also demonstrate empirical analysis of the structures using real-world datasets (3 million records) and also 
	with synthetic datasets (1 billion records) capturing various input distributions. 
\end{abstract}


\end{frontmatter}



\section{Introduction and Motivation}
\label{sec:intro}

Data intensive computing has emerged as a central research theme in the databases and data streams community. With the tremendous spurt in the amount of data generated 
across varied applications, such as information retrieval, online transaction records, telecommunication call data records (CDR), virus databases, climate warning systems, 
web-pages and medical records to name a few, efficiently processing and managing such huge store of data has become a necessity. The problem is further compounded by the 
presence of spurious duplicates or redundant informations, leading to wastage to precious store space and compute efficiency. Hence, removal of such duplicates help to 
improve the resource utilization and compute power especially in the context of data streams requiring real-time processing at 1 GB/sec or even higher. In this work, we 
propose efficient algorithms to tackle the problem of real-time elimination of duplicate records present in large streaming applications. Formally, this is referred to 
as the \emph{data de-duplication} or \emph{Intelligent Compression} problem, and we use the terms interchangeably.

A national telecommunication network generates call data records, (CDR) storing important information such as the callee number, caller number, duration, etc., for future 
utility. However, redundant or duplicate records may be generated due to errors in the procedure. Storing of such billions of CDR in real-time in the central data repository 
calls for duplicate detection and removal to enhance performance. Typical approaches involving the use of database queries or Bloom Filter~\cite{20} are prohibitively slow or are 
extremely resource intensive requiring around 20 GB for storing 6 billion CDR. Even disk-based algorithms have a heavy performance impact. Hence, there is a paramount need 
for deduplication algorithms involving in-memory operations, real-time performance along with tolerable false positive, (FP) and false negative, (FN) rates. 

The growth of search engines provide another field of application for the deduplication algorithms. The search engines need to regularly crawl the Web to extract new URLs 
and update their corpus. Given, a list of extracted URLs, the search engines needs to perform a probe of its corpus to identify if the current URL is already present in its 
corpus~\cite{24}. This calls for efficient duplicate detection, wherein a small performance hit can be tolerated. A high FNR, leading to recrawling of a URL, will lead to a 
severe performance degradation of the search engine and a high FPR, leading to new URLs being ignored, will produce a stale corpus. Hence a balance in both FPR and FNR needs 
to be targeted.

Another interesting application for approximate duplicate detection in streaming environment is the detection of fraudulent advertiser clicks~\cite{31}. In web advertising domain, 
for the sake of profit it is possible that the publisher fakes a certain amount of the clicks (using scripts). The advertising commission necessarily need to detection 
such malpractices. Detection of same user ID or click generation IP in these cases can help minimizing frauds.

Straight-forward approaches to tackle this problem involving pair-wise string comparisons leads to quadratic time complexity prohibiting real-time performance. To address this 
issue, Bloom Filter are typically used in such domains. However, this involves huge memory requirements for tolerable performance of the algorithms and hence led to disk-based 
Bloom Filter approaches which again suffers from reduced throughput due to disk access overhead.

In order to address these challenges, we present the design of novel Bloom Filter based algorithms based on biased sampling, Reservoir sampling, and load based sampling. We 
theoretically analyze the performance of our algorithms and prove it to outperform the competing methods. We also show exhaustive empirical results to validate the enhanced 
performance of our methods in real-time. Using huge datasets of the order of billions of records, we portray better FPR, FNR and convergence to stability of the algorithms.

In the next section we present the related work and existing methods in this problem domain, and discuss the various techniques used in this work. Section~\ref{sec:rsbf} 
presents the working details of the \emph{Reservoir Sampling based Bloom Filter} algorithm, wherein we provide a novel hybrid approach based on Reservoir sampling coupled 
with Bloom Filter. To the best our knowledge this is the first such attempt at combining the two for deduplication applications. It is followed by the biased sampling based 
techniques, \emph{Biased Sampling based Bloom Filter} approaches, where biased sampling functions are used to operate on the Bloom Filters. Section~\ref{sec:rlbsbf} presents a 
randomized load balanced approach involving the load of each Bloom Filter to model its response towards each input element. We next present detailed experimental results on 
both real and synthetic datasets to exhibit the efficient performance of the proposed techniques. Finally, Section~\ref{sec:conc} concludes the work and provide possible future 
direction of work in this area.

\section{Background and Related Work}
\label{sec:rel}

Duplicate detection provides a classical problem within the ambit of data storage and databases giving rise to numerous buffering solutions. 
The advent of online arrival of data and transactions, detection of duplicates in such streaming environment using buffering and caching 
mechanisms~\cite{19} corresponds to a na\"ive solution given the inability to store all the data arriving on the stream. This led to the 
design of fuzzy duplicate detection mechanisms~\cite{8,41}.

Management of large data streams for computing approximate frequency moments~\cite{5}, element classification~\cite{23}, correlated aggregate 
queries~\cite{21} and others with limited memory and acceptable error rates have become a spotlight among the research community. \emph{Bit 
Shaving}, the problem of fraudulent advertisers not paying commission for a certain amount of the traffic or hits have also been studied in 
this context~\cite{36}. This prompted the growth of approximate duplicate detection techniques in the area of both databases and web applications. 
Redundancy removal algorithms for search engines were first studied in~\cite{11,12,29}. File-level hashing was used in storage systems to 
help detect duplicates~\cite{1,15,39}, but they provided a low compression ratio. Even secure hashes were proposed for fixed-sized data blocks~\cite{35}.

Bloom Filters were first used by TAPER system~\cite{27}. A Bloom Filter is a space-efficient probabilistic bit-vector data structure that is widely used for 
membership queries on sets~\cite{9}. Typical Bloom Filter approaches involve $k$ comparisons for each record, where $k$ is the number of hash functions 
used per record for checking the corresponding bit positions of the Bloom Filter array. However, the efficiency of Bloom Filters come at the cost of 
a small false positive rate, wherein the Bloom Filter falsely reports the presence of the query element. This occurs due to hash collision of multiple 
elements onto a single bit position of the Bloom Filter. However, there is no false negative. The probability of false positive for a standard Bloom 
Filter is given by~\cite{10}:
\begin{align}
	FPR \approx \left(1-e^{-kn/m}\right)^k \nonumber
\end{align}
Given $n$ and $m$, the optimal number of hash functions $k = \ln 2(m/n)$.

Counting Bloom Filters~\cite{16} were introduced to support the scenario where the contents of a set change over time, due to insertions and deletions. 
In this approach the bits were replaced by small counters which were updated with the insert and delete of elements. However, the support for deletion 
operations from the structure gave rise to false negatives, where an element was wrongly reported as absent from the set. To meet the needs of varied 
application scenarios, a large number of Bloom Filter variants were proposed such as the compressed Bloom Filter~\cite{32}, space-code Bloom Filter~\cite{28}, 
and spectral Bloom Filter~\cite{37} to name a few. Even window model of Bloom Filters were proposed~\cite{31} such as landmark window, jumping window, 
sliding window~\cite{38}, etc. These models operated on a definite amount of history of objects observed in the stream to draw conclusions for processing 
of future elements of the stream. Parallel variants of Bloom Filters were also explored.

Bloom Filters have been applied even to network related applications such as finding heavy flows for stochastically fair blue queue management~\cite{17}, 
packet classification~\cite{7}, per-flow state management and longest prefix matching~\cite{14}. Multiple Bloom Filters in conjunction with hash tables 
have been studied to represent items with multiple attributes accurately and efficiently with low false positive rates~\cite{26}. \emph{Bloomjoin} used 
for distributed joins have also been extended to minimize network usage for query execution based on database statistics. Bloom Filters have also been 
used for speeding up name-to-location resolution process~\cite{30}.

An interesting Bloom Filter structure proposed recently is the \emph{Stable Bloom Filter}, SBF~\cite{13}. It provides a stable performance guarantee on 
a very large stream. This constant performance is of huge importance for de-duplication applications. SBF works by continuously evicting stale information 
from the Bloom Filters. Although it achieves a tight upper bound on FPR, the stability of the algorithm is reached theoretically at infinite stream length. 
In this work we present a combination of Bloom Filter and Reservoir sampling and show that the proposed method provides lower FNR, comparable FPR, but above 
all converges to stability much faster as compared to SBF.

Finding the number of distinct elements in a stream was explored in~\cite{18}. The problem of synopsis maintenance~\cite{6,22} has been studied in great 
detail for its extensive application in query estimation~\cite{33}. Many synopsis methods such as sampling, wavelets, histograms and sketches have been 
designed for approximate query answering. A comprehensive survey of stream synopsis methods can be found in~\cite{2}. An important class of synopsis 
construction methods is the \emph{Reservoir sampling}~\cite{40}. This sampling method has great appeal as it generates a sample of original multi-dimensional 
data and can be used with various data mining applications.

In Reservoir sampling one maintains a reservoir of size $n$ from the data stream. After the first $n$ points have been added to the reservoir, subsequent 
elements are inserted into the reservoir with an \emph{insertion probability} given by $n /t$ for the $t^{th}$ element of the stream. An interesting 
characteristic of this algorithm is that it is extremely easy to implement and that all subsets of data are equi-probable to be present in the 
reservoir. Each data point is also associated with a bias function representing its probability to be inserted into the reservoir. Hence, the 
procedure can inherently capture changing behavior of the stream with different such biasing functions.

A memory-less temporal bias functions for streams for evolving streams have been proposed in~\cite{3}. Apart from $O(1)$ processing time per stream element, 
incorporating the bias results in upper bounds of reservoir sizes limiting the maximum space requirement to nearly constant in most cases even for an 
infinitely long data stream. In this work we present several biased sampling techniques on the Bloom Filters, and also propose a randomized load balanced 
biasing scheme for the de-duplication problem.

\section{Reservoir Sampling based Bloom Filter (RSBF) Approach }
\label{sec:rsbf}

In this section, we propose the design and working model of the \emph{Reservoir Sampling based Bloom Filter} (RSBF) for de-duplication in large data streams. 
$RSBF$ intelligently combines the concepts of reservoir sampling techniques~\cite{42} and that of Bloom Filter approach. To the best of our knowledge, such an 
integration has not been proposed so far in the literature.

RSBF comprises $k$ Bloom Filters, each of size $s$ bits and are initially set to $0$. On arrival of a new element, $e$ it is hashed to one of the $s$ bits in each of the 
$k$ Bloom Filters with the help of $k$ different uniform random hash functions. The existence of the element is verified by checking whether these $k$ bit 
positions are set. If all the $k$ bit positions are set to $1$, then RSBF reports the element to be duplicate, else to be distinct.
RSBF directly inserts the initial $s$ elements of the stream into the structure by setting the corresponding $k$ bit positions in the Bloom Filter. Each element 
$e_i$, for $i>s$, is then first probed against the Bloom Filter structure to determine the duplicate or distinct status. If $e_i$ is reported as distinct, it is inserted 
in the structure with probability $p_i=s/i$ (insert probability) where $i$ is the current length of the stream and $s$ is the the size of each of the Bloom filter.

However, with the increase in the number of bits set in the Bloom Filters, RSBF would suffer from a high rate of \emph{false positives} wherein a distinct element is falsely 
reported as duplicate. As the length of the 
stream increases, it can be observed that the probability of an element being a duplicate increases (since the elements are drawn from a finite universe). The reservoir 
sampling method implicitly helps to prevent such a scenario by increasingly rejecting elements from being inserted into the structure (as the \emph{insert probability} 
decreases). Insertion of elements from a possibly infinite stream would inevitable lead to the setting of nearly all the bits of RSBF to $1$, thereby incurring a high 
false positive rate (FPR). To alleviate this problem, whenever an element is inserted into RSBF, the algorithm also deletes $k$ randomly uniformly chosen bit (one from each 
Bloom Filter) by setting it to $0$. It should be observed that such deletion operation invariables leads to the presence of \emph{false negatives}, where a duplicate 
element is reported as distinct.

Applications involving duplicate detection demand low tolerance for both false positive as well as false negative rates (FNR). We observed that the use of reservoir sampling helps to keep 
the false positive rate significantly lower. However, the repeated rejection of elements (possibly distinct) with increase in the stream length may result in an increase of the FNR, thereby 
degrading the performance of RSBF. In order to address this problem, we introduce a weak form of biasing on the reservoir sampling operation performed on the stream elements. 
When the insert probability of an element decreases beyond a specified threshold, $p^*$ and is reported as distinct by probing its bits, the element is inserted. This novel combination 
of reservoir sampling with thresholding thus helps to reduce FNR to acceptable limits. This procedure also helps RSBF to dynamically adapt itself to an evolving stream.

We emphasize that along with observing a low FPR and FNR, RSBF also exhibits faster convergence to stability, as compared to that of SBF, as the setting and deletion of $k$ bits 
lead to a near constant number of 1's and 0's in the structure. The pseudo-code for the working of RSBF is presented in Algorithm~\ref{algo:rsbf} and its structure is diagrammatically 
represented by Fig.~\ref{fig:rsbf}. In the following section, we provide a detailed theoretical analysis of RSBF, and later provide empirical results to validate our claims.

\begin{algorithm}[H]
\begin{center}
        \caption{$RSBF(S)$}
        \label{algo:rsbf}
        \begin{algorithmic}
                \REQUIRE Threshold FPR ($FPR_t$), Memory in bits ($M$), and Stream ($S$)
                \ENSURE Detecting \emph{duplicate} and \emph{distinct} elements in $S$

                \medskip

                \STATE Compute the value of $k$ from $FPR_t$.
                \STATE Construct $k$ Bloom filters each having $M/k$  bits of memory.
                \STATE $iter \gets 1$ 
                \FOR{each element $e$ of $S$}
                	\STATE Hash $e$ into $k$ bit positions, $H={h_1,\cdots,h_k}$.
                    \IF{all bit positions in $H$ are set}
                    	\STATE $Result \gets DUPLICATE$
					\ELSE
						\STATE $Result \gets DISTINCT$
					\ENDIF					                   
                    \IF{ $iter \leq s$ }
                    	\STATE Set all the bit positions in $H$.
                    \ELSE
                    	\IF{$(s/iter) \leq p^{*}$}
                    		\FORALL{positions $h_i$ in $H$}
                    			\IF{$h_i=0$}
                    				\STATE Find a bit in $i^{th}$ bloom filter which is set to 1, and reset to 0.
                                    \STATE Set the bit at $h_i$ position to 1
                    			\ENDIF
                    		\ENDFOR
                    	\ELSE
                    	\STATE With probability $(s/iter)$ insert $e$ by setting all the bit positions in $H$.
                    	\STATE If $e$ was decided to be inserted then randomly reset one bit positions from each of the $k$ Bloom filters.
						\ENDIF
					\ENDIF
					\STATE $iter \gets iter +1$
                \ENDFOR

        \end{algorithmic}
\end{center}
\end{algorithm}

\begin{figure}[htpb]
\begin{center}
	\includegraphics[width=\columnwidth]{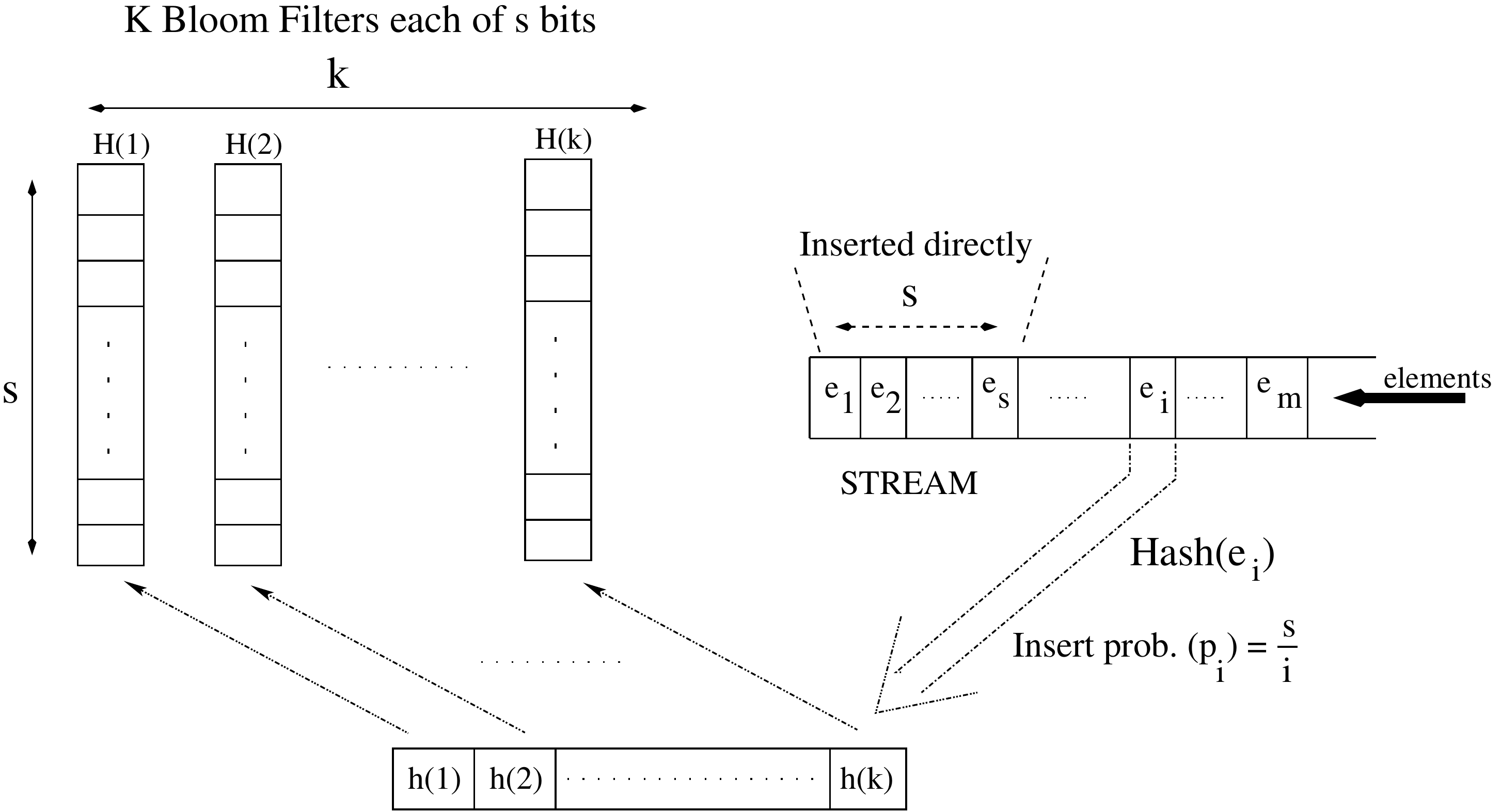}
	\caption{The working model of RSBF.}
	\label{fig:rsbf}
\end{center}
\end{figure}

\subsection{General Framework}
\label{subsec:framework}

In this section we present a generic framework for analyzing the false positive rate (FPR) and the false negative rate (FNR) of our proposed Bloom Filter based algorithms.

The event of a false positive (FP) occurs when a distinct element of the stream is reported as duplicate. A false negative (FN) event occurs when a duplicate element of the stream is reported as distinct.
Now we consider the scenarios under which FP or FN can take place. Assume $e_{m+1}$, the $(m+1)^{th}$ element of the stream to have arrived, and is hashed to $H_{m+1}={h_1,h_2,...,h_k}$ positions where 
$h_i \in [1,s]$ for the $i^{th}$ Bloom Filter. $e_{m+1}$ will be reported as a duplicate if all the bit positions in $H_{m+1}$ are already set to $1$. Let $X_{m+1}$ be the probability of this event. 
If at least one of the bit positions in $H_{m+1}$ is $0$, then $e_{m+1}$ will be reported as distinct.  Also, let us denote by $Y_{m+1}$ the probability that $e_{m+1}$ is actually a distinct element. 
Hence, we have
\begin{align}
X_{m+1} &= P(\text{all bit positions in $H_{m+1}$ are $1$ when $e_{m+1}$ arrived})\\
Y_{m+1} &= P(\text{$e_{m+1}$ is actually a distinct element})
\end{align}
Assume $FPR_{m+1}$ and $FNR_{m+1}$ denotes the probability that ${m+1}^{th}$ element of the stream is a FP and FN respectively, which we show to be determined by the quantities $X_{m+1}$ and $Y_{m+1}$. 
So,
\begin{align}
&FPR_{m+1}= P(\text{$e_{m+1}$ is actually distinct}).P(\text{$e_{m+1}$ is reported duplicate}) \nonumber\\
&= P(\text{$e_{m+1}$ is actually distinct}).P(\text{all bit positions in $H_{m+1}$} \nonumber \\
& \qquad \qquad\text{are $1$ when $e_{m+1}$ arrived}) \nonumber \\
&= Y_{m+1}.X_{m+1} \\		
&FNR_{m+1}= P(\text{$e_{m+1}$ is actually a duplicate}).P(\text{$e_{m+1}$ is reported distinct}) \nonumber\\
&= P(\text{$e_{m+1}$ is actually a duplicate}).P(\text{not all bit positions in} \nonumber \\
& \qquad \qquad \text{$H_{m+1}$ are $1$ when $e_{m+1}$ arrived}) \nonumber\\
&= (1-Y_{m+1}).(1-X_{m+1})
\end{align}
For simplicity of analysis, we consider the probability of two more events: (i) the algorithm correctly predicts $e_{m+1}$ as duplicate, (ii) the algorithm correctly predicts $e_{m+1}$ as distinct. 
Let us denote by $DUP_{m+1}$ and $DIS_{m+1}$ the probability of the event (i) and (ii) respectively. Thus, the expression for $DUP_{m+1}$ and $DIS_{m+1}$ are,
\begin{align}
&DUP_{m+1}= P(\text{$e_{m+1}$ is actually duplicate}).P(\text{$e_{m+1}$ is reported duplicate}) \nonumber\\
&= P(\text{$e_{m+1}$ is actually duplicate}).P(\text{all bit positions in $H_{m+1}$} \nonumber \\
& \qquad \qquad\text{are $1$ when $e_{m+1}$ arrived}) \nonumber \\
&= (1-Y_{m+1}).X_{m+1}
\end{align}
\begin{align}		
&DIS_{m+1}= P(\text{$e_{m+1}$ is actually a distinct}).P(\text{$e_{m+1}$ is reported distinct}) \nonumber\\
&= P(\text{$e_{m+1}$ is actually a distinct}).P(\text{not all bit positions in} \nonumber \\
& \qquad \qquad \text{$H_{m+1}$ are $1$ when $e_{m+1}$ arrived}) \nonumber\\
&= Y_{m+1}.(1-X_{m+1})
\end{align}
We assume that elements of the stream are uniformly randomly drawn from a finite universe $\Gamma$, with $|\Gamma| = U$. Then the probability that $e_{m+1}$ is indeed a distinct element, $Y_{m+1}$, is,
\begin{align}
Y_{m+1} = \left(\frac{U-1}{U}\right)^m
\end{align}
It can be observed that as $U-1 \leq U$, $Y_{m+1}$ tends to $0$ as the stream length, $m$ tends to infinity. Hence inherently the FPR for the algorithms tend to $0$ with increase in the stream length. 
This can be attributed to the fact that the probability of an incoming element to be distinct decreases with stream length as the elements of the stream are drawn from a finite universe. 
However to ensure low FNR for large streams, $X_{m+1}$ must tend to $1$ as $m$ increases. We later show that this property is maintained by our proposed algorithms, thereby attaining a very low FNR as well 
as very low FPR on large data streams.

\subsection{Analysis of RSBF}
\label{subsec:anal_rsbf}

A detailed analysis of the FPR, FNR, and convergence rate of RSBF is provided in our previous work~\cite{42}. However, the impact of $p^*$ was left as future work in that contribution. We now 
provide the complete analysis of RSBF including the factor $p^*$. Recall that for RSBF without $p^{*}$, we have
\begin{align}
\label{eq:FPR}
FPR_{m+1}&={\left(\frac{U-1}{U}\right)}^{m}.{\left[1-\frac{ks}{m}+{\left(\left[1-\frac{1}{e}\right].\frac{s}{m}\right)}^{k}\right]} \\
\label{eq:FNR}
FNR_{m+1}&\approx O\left(\frac{k}{U}\right)
\end{align}
From equation~\eqref{eq:FPR}, we observe that the right multiplicative factor tends to $1$ as the stream length $m$ reaches infinity. However, the left multiplicative factor tends to $0$ as $U-1<U$. 
Hence with increasing stream length, the FPR decreases and nearly becomes constant. From equation~\eqref{eq:FNR} we observe that the FNR becomes constant as the stream length increases. We next present 
the analysis for the RSBF in the presence of $p^*$ and show that RSBF with $p^*$ also exhibits same property. We use the generic framework discussed earlier to perform the analysis.

We now derive the expression for $X_{m+1}$ and use it for computing $FPR_{m+1}$ and $FNR_{m+1}$. Assume that $e_{m+1}$, the $(m+1)^{th}$ element of the stream hashes to $H_{m+1}={h_1,h_2,...,h_k}$ 
positions where each $h_i \in [1,s]$ for the $i^{th}$ Bloom filter. Since all the Bloom filters are identically processed, we perform the analysis for one Bloom filter and then extends the analysis for $k$ Bloom Filters.
Assume $l$ is the last iteration when $h_{i}$ bit position of the $i^{th}$ Bloom filter made a transition from $0$ to $1$ and thereafter it was never reset. The bit position $h_{i}$ will not be reset after 
the $l^{th}$ iteration under the following two conditions: \\
(i) $e_{j}(l+1\leq j\leq m)$ is not inserted by RSBF, or, \\
(ii) $e_{j}(l+1\leq j\leq m)$ is selected for insertion but some bit position other than $h_{i}$ is selected for deletion from $i^{th}$ Bloom filter.

We now define three different sequential phases in the RSBF algorithm: (i) Phase 1: During this phase all the elements of the stream is inserted into the Bloom filter structure. This phase continues 
to run till $s^th$ element of the stream is processed. (ii) Phase 2: This phase starts after Phase 1 and during this phase distinct reported elements are inserted into the Bloom filter with probability 
$s/i$ where $i$ is the current iteration. This phase ends when $p^{*}$ starts operating. (iii) Phase 3: During this phase the insertion probability of an element fall below $p^{*}$ and hence the distinct 
elements are always inserted into the Bloom Filters. We denote the position of the stream from which the effect of $p*$ starts taking place by $p$, and assume $m+1 > p$. 

$h_i$ is set to $1$ for the last time in the $l^{th}$ iteration and after that it is not reset. We denote the probability of this event by $P_{trans}^{l}$. Also we denote by $NR_{l}^{p}$ the probability 
that $h_i$ was not reset from iteration $l$ to $p$ when $l>s$. Let $NR_{l}^{m}$ capture the probability that $h_i$ is not reset in iterations from $l$ to $m$ when $l>p$. $l$ can be in one of the three phases 
that we have defined above.

If $l$ lies in the Phase 1, ($1\leq l\leq s$)
\begin{align}
P_{trans}^{l}&=P(\text{$e_l$ chooses $h_i$}).P(\text{$h_i$ was never reset thereafter}) \\ \nonumber
&=P(\text{$e_l$ chooses $h_i$}).NR_{s}^{p}.NR_{p}^{m}
\end{align}

If $l$ lies in the Phase 2, ($s\leq l\leq p$)
\begin{align}
P_{trans}^{l}&=P(\text{$e_l$ is reported distinct}).P(\text{$e_l$ is inserted}).P(\text{$e_l$ chooses $h_i$}).\\ \nonumber
&P(\text{$h_i$ was never reset thereafter}) \\ \nonumber
&=(1-X_l)\left(\frac{s}{l}\right)\left(\frac{1}{s}\right) .NR_{l}^{p}.NR_{p}^{m} \\ \nonumber
&=(1-X_l)\left(\frac{1}{l}\right).NR_{l}^{p}.NR_{p}^{m}
\end{align}

If $l$ lies in the Phase 3, ($p\leq l\leq m$)
\begin{align}
P_{trans}^{l}&=P(\text{$e_l$ is reported distinct}).\\ \nonumber
&P(\text{$e_l$ chooses $h_i$}).P(\text{$h_i$ was never reset thereafter}) \\ \nonumber
&=(1-X_l)\left(\frac{1}{s}\right).NR_{l}^{m}
\end{align}

We now derive the expressions for $NR_{s}^{p}$.
\begin{align}
NR_{s}^{p}&=P(\text{$h_i$ was not reset after it was set before $s^{th}$ iteration}) \\ \nonumber
&=\prod_{i=s+1}^{p} P(\text{$e_i$ was reported duplicate}) + P(\text{$e_i$ was reported distinct}). \\ \nonumber
&[P(\text{$e_i$ was not inserted}) + P(\text{$e_i$ was inserted}) \\ \nonumber
&P(\text{Some other bit position than $h_i$ was selected for deletion})] \\ \nonumber
&=\prod_{i=s+1}^{p} \left(X_i+\left(1-X_i\right)\left( \left( 1-\frac{s}{i}\right)+ \frac{s}{i}.\frac{s-1}{s}\right)\right) \\ \nonumber
&=\prod_{i=s+1}^{p}\left( X_i + \left(1-X_i\right)\left(1-\frac{1}{i}\right)\right)
\end{align}

Similarly $NR_{l}^{p}$ where $s\leq l < p$ is as follows,
5
\begin{align}
NR_{l}^{p}&=\prod_{i=l+1}^{p}\left( X_i + \left(1-X_i\right)\left(1-\frac{1}{i}\right)\right)
\end{align}

Expressions for $NR_{p}^{m}$ and $NR_{l}^{m}$ are similarly derived for $p\leq l < m$.
\begin{align}
NR_{p}^{m}&=\prod_{i=p+1}^{m}\left( X_i + \left(1-X_i\right)\left(1-\frac{1}{s}\right)\right) \\
NR_{l}^{m}&=\prod_{i=l+1}^{m}\left( X_i + \left(1-X_i\right)\left(1-\frac{1}{s}\right)\right)
\end{align}

It should be noted that after crossing the point $p$, all the distinctly reported data points are inserted and each time an insertion takes place, deletion occurs.
We now derive the expression for $X_{m+1}$. Since the value of $l$ can vary, we sum over all possible values of $l$ in the different ranges. The probability that $h_i$ was set in during the 
first $s$ iterations is given by $\left\lbrace 1-{\left(1-\frac{1}{s}\right)}^s\right\rbrace$. Therefore we get, for ($(m+1) > p$)
\begin{align}
X_{m+1}&={\left[A_{1-s}+A_{s-p}+A_{p-m}\right]}^{k}
\end{align}
where
\begin{align}
A_{1-s}=&\left\lbrace 1 - {\left(1-\frac{1}{s}\right)}^{s}\right\rbrace \left\lbrace NR_{s}^{p} \right\rbrace \left\lbrace NR_{p}^{m} \right\rbrace \\
A_{s-p}=&\sum_{l=s+1}^{p}\left\lbrace \left(1 - X_{l}\right)\frac{1}{l}\right\rbrace \left\lbrace NR_{l}^{p} \right\rbrace \left\lbrace NR_{p}^{m} \right\rbrace \\ 
A_{p-m}=&\sum_{l=p+1}^{m}\left\lbrace\left(1-X_l\right).\frac{1}{s}\right\rbrace \left\lbrace NR_{l}^{m} \right\rbrace
\end{align}

The probability of an FNR at point $m+1$ where $(m+1) < p$ follows directly from the above expression,
\begin{align}
\label{eq:X_m_l_p}
X_{m+1}&={\left[A_{1-s}^{'}+A_{s-m}^{'}\right]}^{k}
\end{align}
where
\begin{align}
A_{1-s}^{'}=&\left\lbrace 1 - {\left(1-\frac{1}{s}\right)}^{s}\right\rbrace \left\lbrace NR_{s}^{m'} \right\rbrace \\
A_{s-m}^{'}=&\sum_{l=s+1}^{p}\left\lbrace \left(1 - X_{l}\right)\frac{1}{l}\right\rbrace \left\lbrace NR_{l}^{m'} \right\rbrace \\ 
\end{align}

The terms $NR_{s}^{m'}$ and $NR_{l}^{m'}$ are obtained by carrying out similar analysis as before.
\begin{align}
NR_{s}^{m'}&=\prod_{i=s+1}^{m}\left( X_i + \left(1-X_i\right)\left(1-\frac{1}{i}\right)\right) \\
NR_{ l}^{m'}&=\prod_{i=l+1}^{m}\left( X_i + \left(1-X_i\right)\left(1-\frac{1}{i}\right)\right)
\end{align}

Careful observation of the expression for $X_{m+1}$ provides the following recurrence relation,
If $m\leq p$,
\begin{align}
\label{eq:X_m_leq_p}
X_{m+1}&={\left[ {\left(X_m\right)}^{\frac{1}{k}} \left\lbrace X_m + \left( 1-X_m \right) \left( 1-\frac{1}{m} \right) \right\rbrace + \left( 1-X_m \right).\frac{1}{m} \right]}^k 
\end{align}

Else ($m>p$)
\begin{align}
\label{eq:X_m_geq_p}
X_{m+1}&={\left[ {\left(X_m\right)}^{\frac{1}{k}} \left\lbrace X_m + \left( 1-X_m \right) \left( 1-\frac{1}{s} \right) \right\rbrace + \left( 1-X_m \right).\frac{1}{s} \right]}^k 
\end{align}

In the following lemma we show that $X$ is monotonically increasing and converges to $1$. As $FNR_{m+1} = (1-Y_{m+1}).(1-X_{m+1})$~\ref{subsec:framework}, $RSBF$ exhibits very low FNR with 
increase in stream length. We later present empirical results to validate the claim that $X$ converges to $1$ along with a fast convergence rate.
\begin{theorem}
\label{lem:RSBF_X}
For RSBF $X$ monotonically increases and converges to $1$. Therefore FNR tends to $0$ with increase in stream length.
\end{theorem}
\begin{proof}
From Eq.~\eqref{eq:X_m_l_p} and Eq.~\eqref{eq:X_m_leq_p} we observe that, $X_1 = 0$ and $X_2 = \frac{1}{m^k}$. Hence, using Eq.~\eqref{eq:X_m_leq_p}  and Eq.~\eqref{eq:X_m_geq_p} we have,
For $m \leq p$
\begin{align}
{\left(\frac{X_{m+1}}{X_m}\right)}^{\frac{1}{k}}&=X_{m}+(1-X_{m})(1-\frac{1}{m})+\frac{1}{X_{m}^{\frac{1}{k}}}.(1-X_{m}).\frac{1}{m} \nonumber \\
&=1-(1-X_m).\frac{1}{m}+(1-X_m).\frac{1}{m.X_{m}^{\frac{1}{k}}} \nonumber \\
\label{eq:ratio_m_leq_p}
&=1+(1-X_m).\frac{1}{m}.(X_{m}^{-\frac{1}{k}}-1)
\end{align}
For $m>p$
\begin{align}
{\left(\frac{X_{m+1}}{X_m}\right)}^{\frac{1}{k}}&=X_{m}+(1-X_{m})(1-\frac{1}{s})+\frac{1}{X_{m}^{\frac{1}{k}}}.(1-X_{m}).\frac{1}{s} \nonumber \\
&=1-(1-X_m).\frac{1}{s}+(1-X_m).\frac{1}{s.X_{m}^{\frac{1}{k}}} \nonumber \\
\label{eq:ratio_m_ge_p}
&=1+(1-X_m).\frac{1}{s}.(X_{m}^{-\frac{1}{k}}-1)
\end{align}
We observe that the right hand side of both Eq.~\eqref{eq:ratio_m_leq_p} and  Eq.~\eqref{eq:ratio_m_ge_p} is greater than $1$ when $X_m<1$. Therefore for $X_m \leq 1$, $X_{m+1} \geq X_{m}$ 
where the equality holds only if $X_m=1$. Hence $X$ monotonically increases and converges to $1$. As such, FNR tends to $0$ as the stream length increases.
\end{proof}

\section{Biased Sampling based Bloom Filter (BSBF) and its Variants}
\label{sec:bsbfv}

In this section, we propose variants of the Bloom Filter approach for de-duplication purposes. We put forth several versions of the biased sampling techniques with detailed analysis of their 
performance. We further show that these structures can efficiently handle evolving data streams, providing improved performance over the state-of-art. Later we exhibit empirical 
results to validate our claims. It can be observed that the various modifications discussed are modelled to handle various distributions of the input stream.

The \emph{Biased Sampling based Bloom Filter} (BSBF) approach works on similar lines as that of RSBF, albeit with a small variation on the insertion criteria. 
BSBF inserts into its structure all the elements arriving on the stream and reported as distinct, whereas the RSBF in contrast uses the insert probability based on reservoir sampling. 
The insertion follows the same steps as that of RSBF, i.e., $k$ bits are set to $1$, one from each Bloom Filter, based on $k$ uniform hash functions. 
As discussed in the earlier sections, the unbounded insertion of elements leads to an increase in the FPR of the structure. To tackle this problem, BSBF on every 
insertion deletes $k$ randomly uniformly selected bits, one from each of the $k$ Bloom Filters, thereby leading to a balance between the FPR and FNR encountered. 
This deletion procedure is the same as that followed by RSBF. It must be noted that the chosen bit for deletion might have already been set to $0$. As such, we argue 
in similar lines as that of RSBF for nearly a constant number of 1's and 0's in the BSBF structure. Hence, BSBF also exhibits the attractive property of faster convergence 
to stability, just like that of RSBF.

Since the insertion of elements in the BSBF is not restricted, with change in the nature of the input stream, BSBF also updates the bit signature of the elements stored. 
Hence, we can observe that BSBF implicitly captures the biased nature of the stream and dynamically adapts itself. We next present the theoretical analysis for the performance 
of BSBF. Algorithm~\ref{algo:bsbf1} depicts the pseudo-code of the working of BSBF.

\begin{algorithm}[H]
\begin{center}
        \caption{$BSBF(S)$}
        \label{algo:bsbf1}
        \begin{algorithmic}
                \REQUIRE Threshold FPR ($FPR_t$), Memory in bits ($M$), and Stream ($S$)
                \ENSURE Detecting \emph{duplicate} and \emph{distinct} elements in $S$

                \medskip

                \STATE Compute the value of $k$ from $FPR_t$.
                \STATE Construct $k$ Bloom filters each having $M/k$  bits of memory.
                \FOR{each element $e$ of $S$}
                        \STATE Hash $e$ into $k$ bit positions, $H={h_1,\cdots,h_k}$.
                        \IF{all bit at positions $H$ are set}
                        	\STATE $Result \gets DISTINCT$
						\ELSE
							\STATE $Result \gets DUPLICATE$
						\ENDIF

                        \IF{$e$ is DISTINCT}                      
				\STATE Randomly select k bit positions $hat{H}=\hat{h}_{1},\hat{h}_{2},...,\hat{h}_{k}$ one each from the $k$ Bloom filters.
				\STATE Reset all bits in $\hat{H}$ to $0$.	
				\STATE Set all the bits in $H$ to $1$.
			\ENDIF
                \ENDFOR
        \end{algorithmic}
\end{center}
\end{algorithm}

\subsection{Analysis of BSBF}
\label{subsec:anal_bsbf}

The probability of FPR or FNR at $e_{m+1}$ depends on the value of $X_{m+1}$, the probability that all bit positions in $H_{m+1}$ are $1$ when $e_{m+1}$ arrives. Using the framework discussed earlier in 
section~\ref{subsec:framework} and the values of $FPR_{m+1}$ and $FNR_{m+1}$, we now derive the expression of $X_{m+1}$ for $BSBF$.

The $(m+1)^{th}$ element of the stream, $e_{m+1}$ hashes to $H_{m+1} = {h_1,h_2,...,h_k}$ positions, where each $h_i \in [1,s]$ for the $i^{th}$ Bloom filter. Initially all the bits of the Bloom filters are set to $0$. 
Since all the Bloom filters are identical and independent, we first consider only a single Bloom filter and later extend our arguments for all the $k$ Bloom Filters. We assume $l$ to be the last iteration 
whence the ${h}_{i}^{th}$ bit position of the $i^{th}$ Bloom filter makes the last transition from $0$ to $1$, and thereafter $h_i$ is never reset. A bit in the Bloom filter is reset by BSBF only when 
some element is inserted. Hence, $h_i$ should not be reset in an iteration $j$, ($l+1 \leq j \leq m$) if $e_j$ is not inserted (that is $e_j$ is reported as duplicate) or if $e_j$ is selected for insertion 
and some other bit position is chosen for reset. The probability of some other bit position to be chosen is given by $\left(1-\frac{1}{s}\right)$. We denote the probability of such a transition by $P_{trans}^{(l)}$. Hence,
\begin{align}
&P_{trans}^{(l)}= P(\text{$l$ is the last iteration when $h_{i}$ is set to $1$ thereafter it was never reset}) \nonumber \\
&=P(\text{$e_{l}$ was reported distinct}).P(\text{$e_l$ chooses $h_{i}$}).P(\text{$h_{i}$ was never reset thereafter}) \nonumber\\
&=P(\text{$e_{l}$ was reported distinct}).P(\text{$e_l$ chooses $h_{i}$}) \left\lbrace \prod_{i=l+1}^{m}\left[ P(\text{$e_{i}$ was reported duplicate}) \right. \right.\nonumber \\ 
& \qquad \left. \left. + P(\text{$e_{i}$ was reported distinct}).(1-\frac{1}{s})\right] \right\rbrace \nonumber\\
&=\left\lbrace(1-X_{l}).Y_{l} + (1-X_{l})(1-Y_{l})\right\rbrace.\frac{1}{s}.\left\lbrace \prod_{i=l+1}^{m}\left[X_{i}+(1-X_{i})(1-\frac{1}{s})\right]\right\rbrace
\end{align}
Since this transition can happen in any iteration from $1$ to $m$, $l \in [1,m]$. As the same analysis holds for all the $k$ Bloom filters, we obtain the expression for $X_{m+1}$ as follows,
\begin{align}
\label{eq:X}
X_{m+1}&={\left[\sum_{l=1}^{m}\left\lbrace(1-X_{l}).Y_{l} + (1-X_{l})(1-Y_{l})\right\rbrace.\frac{1}{s}.\left\lbrace \prod_{i=l+1}^{m}\left[X_{i}+(1-X_{i})(1-\frac{1}{s})\right]\right\rbrace\right]}^{k} \nonumber \\
&={\left[\sum_{l=1}^{m}(1-X_{l}).\frac{1}{s}.\left\lbrace \prod_{i=l+1}^{m}\left[X_{i}+(1-X_{i})(1-\frac{1}{s})\right]\right\rbrace\right]}^{k}
\end{align}
Carefully observing the right hand side of the above equation, the following recurrence relation for $X_{m+1}$ holds,
\begin{align}
\label{eq:X_mon}
X_{m+1}&={\left[ {(X_{m})}^{\frac{1}{k}}\left\lbrace X_{m}+(1-X_{m}).(1-\frac{1}{s})\right\rbrace + (1-X_{m}).\frac{1}{s}\right] }^{k}
\end{align}
In the next lemma we show that $X$ is monotonically increasing and converges to $1$. As $FNR_{m+1} = (1-Y_{m+1}).(1-X_{m+1})$~\ref{subsec:framework}, $BSBF$ exhibits very low FNR with increase in stream length. 
We later present empirical result to validate the claim that $X$ converges to $1$ along with a fast convergence rate.
\begin{lemma}
\label{lem:X}
$X$ monotonically increases and converges to $1$. Therefore FNR tends to $0$ with increase in stream length.
\end{lemma}
\begin{proof}
From Eq.~\eqref{eq:X} we observe that, $X_1 = 0$ and $X_2 = \frac{1}{s^k}$. Hence, using Eq.~\eqref{eq:X_mon} we have,
\begin{align}
{\left(\frac{X_{m+1}}{X_m}\right)}^{\frac{1}{k}}&=X_{m}+(1-X_{m})(1-\frac{1}{s})+\frac{1}{X_{m}^{\frac{1}{k}}}.(1-X_{m}).\frac{1}{s} \nonumber \\
&=1-(1-X_m).\frac{1}{s}+(1-X_m).\frac{1}{s.X_{m}^{\frac{1}{k}}} \nonumber \\
\label{eq:ratio}
&=1+(1-X_m).\frac{1}{s}.(X_{m}^{-\frac{1}{k}}-1)
\end{align}
We observe that the right hand side of Eq.~\eqref{eq:ratio} is greater than $1$ when $X_m<1$. Therefore for $X_m \leq 1$, $X_{m+1} \leq X_{m}$. Hence $X$ monotonically increases and converges to $1$. 
As such, FNR tends to $0$ as the stream length increases.
\end{proof}

\subsection{BSBF with Single Deletion}
\label{subsec:bsbfsingle}

One of the major problems of deletion from a Bloom Filter arises from the fact that multiple elements may be mapped to a single bit position. Hence the deletion of that 
bit position (reset to $0$) in practice tends to delete multiple elements from the structure. \emph{Counting Bloom Filters} try to alleviate this problem to a certain 
extend, but with enormous space requirements. The heart of the problem lies in the fact that the history regarding the elements which map to a bit position is not 
stored in such structures. This invariably leads to a higher FNR. This provides the basic motivation for this variant of BSBF, \emph{BSBF with Single Deletion} (BSBFSD).

The working of the BSBFSD is in close similarity with that of BSBF. The insertion procedure follows exactly the same steps as that of BSBF. However, with each insertion 
into the structure, we uniformly randomly select one Bloom Filter from which the bit needs to be deleted. Within this selected Bloom Filter, we again uniformly randomly 
choose a bit position to be set to $0$. Hence BSBFSD performs deletion in a conservative manner so as to preserve as many element signature as possible. This leads to 
an improved performance based on FNR criteria. However, since most of the bits are set to $1$ as the stream length increases, BSBFSD incurs a higher FPR compared to 
BSBF. However, certain application like federal crime record corpuses require extremely low or zero FNR (but can compromise with relatively higher FPR tolerance), 
making BSBFSD particularly suitable for such scenarios. Next we provide theoretical results claiming that the trade-off between FPR and FNR of both the BSBF and 
BSBFSD are better paid-off than that of SBF. The pseudo-code for BSBFSD is provided in Algorithm~\ref{algo:bsbfsd}.

\begin{algorithm}[H]
\begin{center}
        \caption{$BSBFSD(S)$}
        \label{algo:bsbfsd}
        \begin{algorithmic}
                \REQUIRE Threshold FPR ($FPR_t$), Memory in bits ($M$), and Stream ($S$)
                \ENSURE Detecting \emph{duplicate} and \emph{distinct} elements in $S$

                \medskip

                \STATE Compute the value of $k$ from $FPR_t$.
                \STATE Construct $k$ Bloom filters each having $M/k$  bits of memory.
                \FOR{each element $e$ of $S$}
                        \STATE Hash $e$ into $k$ bit positions, $H={h_1,\cdots,h_k}$.
                        \IF{all bit at positions $H$ are set}
                        	\STATE $Result \gets DISTINCT$
						\ELSE
							\STATE $Result \gets DUPLICATE$
						\ENDIF

                        \IF{$e$ is DISTINCT}                      
                        	\STATE Randomly select a Bloom filter $B_{i}$
							\STATE Randomly select a bit $\hat{h}_{i}$ from the $B_{i}^{th}$ Bloom Filter
							\STATE Reset $\hat{h}_{i}$ to $0$.		
							\STATE Set all the bits in $H$ to $1$.
						\ENDIF
                \ENDFOR
        \end{algorithmic}
\end{center}
\end{algorithm}

\subsection{Analysis of BSBFSD}
\label{subsec:anal_bsbsfsd}

It can be observed that the FPR and FNR analysis of BSBFSD is similar to that of BSBF. BSBFSD differs from BSBF in the way deletion of bits take place. In BSBF whenever a new data element 
is inserted into the Bloom filters, one randomly chosen bit was reset from each of the Bloom filters. In BSBFSD, after insertion we randomly select one Bloom filter and then randomly reset 
one bit from the selected Bloom filter. The probability that the bit $h_i$ is not reset in an iteration involving the insertion of some element is $\left( \frac{1}{k}.\left(1-\frac{1}{s}\right) + (1-\frac{1}{k}) \right)$ 
or $\left(1-\frac{1}{ks}\right)$. Hence the expression for $X_{m+1}$ becomes,
\begin{align}
X_{m+1}&={\left[\sum_{l=1}^{m}(1-X_{l}).\frac{1}{s}.\left\lbrace \prod_{i=l+1}^{m}\left[X_{i}+(1-X_{i})(1-\frac{1}{ks})\right]\right\rbrace\right]}^{k} \nonumber \\
&={\left[ {(X_{m})}^{\frac{1}{k}}\left\lbrace X_{m}+(1-X_{m}).(1-\frac{1}{ks})\right\rbrace + (1-X_{m}).\frac{1}{s}\right] }^{k} \nonumber
\end{align}
Similar to the analysis shown in Section~\ref{subsec:anal_bsbf}, $X$ monotonically increases to $1$. The expression for $\frac{X_{m+1}}{X_{m}}$ is,
\begin{align}
{\left(\frac{X_{m+1}}{X_m}\right)}^{\frac{1}{k}}&=X_{m}+(1-X_{m})(1-\frac{1}{ks})+\frac{1}{X_{m}^{\frac{1}{k}}}.(1-X_{m}).\frac{1}{s} \nonumber \\
&=1+(1-X_m).\frac{1}{s}.(X_{m}^{-\frac{1}{k}}-\frac{1}{k})
\end{align}
Using similar arguments from Lemma~\ref{lem:X}, $X_{m+1} \geq X_m$ for $X_m \leq 1$ (monotonically increasing) and converges to $1$.

\section{Randomized Load Balanced Biased Sampling based Bloom Filter}
\label{sec:rlbsbf}

In this section we consider a further variation of the biased sampling based Bloom Filter approach, the \emph{Randomized Load Balanced Biased Sampling based Bloom Filter} (RLBSBF). The main aim 
of this approach is to keep the load (number of 1s) in each of the Bloom Filter below a certain ratio (probabilistically) of its total space, thereby containing the FPR and FNR achieved to a low value.
This would further ensure the stability of the Bloom Filter performance by keeping the count of $0$ and $1$ nearly constant. 

The insertion procedure of an element arriving in the data stream remains the same as that of the other algorithms discussed earlier. That is, when a distinct element arrives, it is 
inserted into the Bloom Filters by appropriately setting the bit positions. However, there is a major difference in the deletion procedure. Whenever an element is inserted, we 
independently access each of the Bloom Filters and based on its current load factor probabilistically decide whether to delete (reset) a random bit or not. This approach intuitively tries to restrict 
FPR to a small value by limiting the number of 1s in each Bloom Filter, along with deleting as less history as possible to obtain a low FNR as well.

Formally, RLBSBF stores the load (the number of bits set) in each of the Bloom Filters to decide the probability of deletion of a bit after the insertion of an element. When an element $e$ is 
reported as distinct by the algorithm, we insert $e$ into the Bloom Filters. Like before, this is done by setting the bits is $H$ to $1$. But instead of performing deterministic deletion from 
the Bloom Filters, we perform randomized deletion from the Bloom Filters. From each filter we first select a random bit position and then reset it with probability $L_{m+1}(i)/s$, where $L_{m+1}(i)$ 
denotes the load of the $i^{th}$ Bloom filter when $e_{m+1}$ arrived. The detailed algorithm is given in Algorithm~\ref{algo:rlbsbf}.

Empirically we observed (results shown later) that this approach in turn enforces a low FPR and the lowest FNR for the competing algorithms. The efficient performance of RLBSBF can be attributed 
to the prevention of the number of set bits in each Bloom filter from becoming high, reducing the FPR, and on the other hand, performing load based deletion of bits (lesser deletion events) curbing down the FNR.

\begin{algorithm}[H]
\begin{center}
        \caption{$RLBSBF(S)$}
        \label{algo:rlbsbf}
        \begin{algorithmic}
                \REQUIRE Threshold FPR ($FPR_t$), Memory in bits ($M$), and Stream ($S$)
                \ENSURE Detecting \emph{duplicate} and \emph{distinct} elements in $S$

                \medskip

                \STATE Compute the value of $k$ from $FPR_t$.
                \STATE Construct $k$ Bloom filters each having $M/k$  bits of memory.
                \FOR{each element $e$ of $S$}
                        \STATE Hash $e$ into $k$ bit positions, $H={h_1,\cdots,h_k}$.
                        \IF{all bit at positions $H$ are set}
                        	\STATE $Result \gets DISTINCT$
						\ELSE
							\STATE $Result \gets DUPLICATE$
						\ENDIF

                        \IF{$e$ is DISTINCT}                      
                        	\FORALL{Bloom filter $B_{i}$}
                        		\STATE Select a random bit position $\hat{h_{i}}$.
                        		\STATE Reset $\hat{h_{i}}$ with probability $L(i)/s$ where $L(i)$ is the number of ones in the Bloom filter $B_{i}$.
                        	\ENDFOR
							\STATE Set all the bits in $H$ to $1$.
						\ENDIF
                \ENDFOR
        \end{algorithmic}
\end{center}
\end{algorithm}

\subsection{Analysis of RLBSBF}
\label{subsec:anal_rlbsbf}

For computing the probability of FNR and FPR of RLBSBF, we initially evaluate the expected load of the Bloom Filters at any iteration. Let $L_{m+1}^{i}$ denotes the expected number of bits 
set in the $i^{th}$ Bloom Filter when element $e_{m+1}$ arrives and is hashed to $h_i$ bit position of the $i^{th}$ Bloom Filter. We are interested in finding out the probability of $h_i$ 
being already set to $1$ for all $i \in [1,k]$. This gives us $X_{m+1}$, which in turn determines the FPR and FNR. 

We assume $l$ to be the last iteration when $h_i$ was set to $1$ and after that it was never reset. $h_i$ will not be reset in an iteration $j (l+1\leq j \leq m)$ if one the following events occur: 
(i) $e_j$ was reported as duplicate and hence not inserted into the Bloom Filter, (ii) $e_j$ was inserted into the Bloom filter and with probability $1-\frac{L_{m+1}^{i}}{s}$ no bits were reset in the 
$i^{th}$ Bloom Filter, (iii) $e_j$ was inserted into the Bloom Filter and with probability $\frac{L_{m+1}^{i}}{s}$ deletion from $i^{th}$ Bloom Filter was performed, but some bit other than $h_i$ was 
chosen for deletion with probability $\left(1-\frac{1}{s}\right)$. It can be observed that this argument applies to all the $k$ bloom filters, and that $l$ can vary from $1$ to $m$.
Therefore, the expression of $X_{m+1}$ for the RLBSBF algorithm is as,
\begin{align}
&X_{m+1}=P(\text{all bit positions in $H_{m+1}$ are $1$ when $e_{m+1}$ arrived}) \nonumber\\
&={\left[\sum_{l=1}^{m}P(\text{$l$ is the last iteration when $h_i$ is set to $1$, thereafter it was never reset}) \right]}^k \nonumber \\
&={\left[\sum_{l=1}^{m}P(\text{$e_l$ was reported distinct}).P(\text{$e_l$ chooses $h_i$}).P(\text{$h_i$ was never reset thereafter}) \right]}^k \nonumber \\
&={\left[\sum_{l=1}^{m}(1-X_{l}).\frac{1}{s}.\left\lbrace \prod_{i=l+1}^{m}\left[X_{i}+(1-X_{i})\left((1-\frac{L_{m+1}^{i}}{s})+\frac{L_{m+1}^{i}}{s}(1-\frac{1}{s})\right)\right]\right\rbrace\right]}^{k} \nonumber \\
&={\left[\sum_{l=1}^{m}(1-X_{l}).\frac{1}{s}.\left\lbrace \prod_{i=l+1}^{m}\left[X_{i}+(1-X_{i})(1-\frac{L_{m+1}^{i}}{s^2})\right]\right\rbrace\right]}^{k} \\
&={\left[ {(X_{m})}^{\frac{1}{k}}\left\lbrace X_{m}+(1-X_{m}).(1-\frac{L_{m+1}^{i}}{s^2})\right\rbrace + (1-X_{m}).\frac{1}{s}\right] }^{k}
\end{align}
We now find the expected load, $L_{m+1}^{i}$ of the $i^{th}$ Bloom Filter at the time when $e_{m+1}$ arrived. Let us associate an indicator random variable $I_j$ with each bit $j$ of the Bloom filters. 
$I_j$ is $1$ if $j=1$, otherwise $0$. Also let $Z^{i}$ be random variable such that $Z^{i}=\sum_{j=1}^{s}I_j$. Therefore $E[Z^{i}]$ will give us the expected load of the $i^{th}$ Bloom filter.
\begin{align}
L_{m+1}^{i}&=E[Z^{i}]\nonumber \\
&=\sum_{j=1}^{s}E[I_j] \nonumber \\
&=\sum_{j=1}^{s}P(\text{$j^{th}$ bit was set to 1 when $e_{m+1}$ arrived}) \nonumber
\end{align}

\section{Experiments and Results}
\label{sec:expt}

In this section we empirically compare the performances of SBF, RSBF, BSBF, BSBFSD and RLBSBF algorithms against various parameters like memory, stream size, percentage of distinct elements 
in the stream etc. We measure the performances on real dataset containing clickstream data (obtained from \url{http://www.sigkdd.org/kddcup/index.php?section=2000&method=data}) having 
around 3M elements as well as on uniformly and randomly generated datasets with upto 1B records. In all experiments $p^*$ for RSBF has been set to 0.03. Also we present experimental results 
for choosing the value of $k$ for BSBF, BSBFSD and RLBSBF algorithms. We show that our proposed algorithms are comparable or outperforms SBF with respect to FPR and FNR while exhibiting better 
stability and faster convergence properties.

\subsection{Setting of Parameters}
\label{subsec:param}

In this section we present the rationale behind setting of the parameter $k$ (the number of Bloom filters) for the various proposed algorithms. Given a fixed memory space $M$, we experimentally 
search for an optimal value of $k$ such that an overall low FPR and FNR is attained.

For RSBF, we have shown in our previous work~\cite{42} that
\begin{align}
\label{RSBF_k}
k=\frac{\ln{({FPR}_t)}}{\ln{\left(1-\frac{1}{e}\right)}}
\end{align}
We also observed that with increase in $k$ FPR decreases, but FNR is minimized when $k=1$. As a trade-off we set $k$ as the arithmetic mean of $1$ and that obtained in ~\eqref{RSBF_k}. The threshold FPR ${FPR}_t$ is set to 0.1.

We now present the performance of BSBF algorithm under various parametric settings of $k$. We chose a uniform random dataset of size 1B with 60\% distinct element. We also vary the memory size from 8MB 
to 512MB and analyze the FPR and FNR.  From Table~\ref{BSBF-I:1B:60distinct} we observe that BSBF exhibits low FNR and high FPR for $k=1$. As we increase the value of $k$, FPR decreases while FNR increases. 
The increase is FNR is attributed to the fact that by increasing the the value of $k$, more element signatures are being deleted from the Bloom Filter structures. (Every time some element is inserted, one bit 
from each of the $k$ Bloom filter is reset to $0$). The decrease in FPR is due to the increase in the signature length of an element in the Bloom Filter. We observe that for $k=2$, both FPR and FNR attains a 
acceptably balanced limit. Hence for performance evaluation of BSBF algorithm, we set $k=2$ for the result of our experimental setup. We also observe that if higher memory space is available, then BSBF attains 
very low FNR (3.4\%) and FPR (6.4\%) for $k=1$. Hence depending on the application specifications BSBF can be modeled to perform efficiently.

\begin{table}[htbp]
\begin{center}
\begin{tabular}{|l|l|l|l|l|l|l|}
\hline
\multicolumn{7}{|c|}{Algorithm:BSBF, Dataset:1B , Distinct:60\% } \\ \hline
Space &  & k=1 & k=2 & k=3 & k=4 & k=5 \\ \hline \hline
\multirow{2}{*}{8 MB} & \% FPR & 75.979 & 35.4924 & 15.2961 & 6.95161 & 3.28054 \\ \cline{2-7}
& \% FNR & 9.20901 & 59.0297 & 82.0828 & 91.5512 & 95.7754\\ \hline \hline
\multirow{2}{*}{128 MB} & \% FPR & 21.0883 & 11.9472 & 6.67181 & 3.47978 & 1.76989 \\ \cline{2-7}
& \% FNR & 9.83893 & 34.7514 & 59.2496 & 75.1262 & 84.1857 \\ \hline \hline
\multirow{2}{*}{512 MB} & \% FPR & 6.46215 & 1.82011 & 0.777613 & 0.386642 & 0.205833 \\ \cline{2-7}
& \% FNR & 3.34108 & 13.5658 & 27.7512 & 43.4681 & 56.9057 \\ \hline \hline
\end{tabular}
\caption{Synthetic Dataset of 1B elements (60\% distinct)}
\label{BSBF-I:1B:60distinct}
\end{center}
\end{table}


We next present the performance of BSBFSD algorithm across various parametric setting of $k$. We use the some random dataset of size 1B with 60\% distinct element and vary the memory size from 8MB to 512MB. 
From Table \ref{BSBFSD:1B:60distinct}, we observe that BSBFSD exhibits increasing FNR and decreasing FPR for increasing $k$ when the memory size is high (128MB and 512MB). The decrease in FPR is due to increase 
in the signature length of an element in the Bloom Filter with increasing $k$. But as a trade-off FNR increases, as chances of more elements being mapped to a point of deletion increases. We observe that for $k=2$, 
both FPR and FNR attains a reasonably balanced limit. Hence we set $k=2$ for BSBFSD algorithm's performance analysis.

\begin{table}[htbp]
\begin{center}
\begin{tabular}{|l|l|l|l|l|l|l|}
\hline
\multicolumn{7}{|c|}{Algorithm:BSBFSD, Dataset:1B , Distinct:60\% } \\ \hline
Space &  & k=1 & k=2 & k=3 & k=4 & k=5 \\ \hline \hline
\multirow{2}{*}{8 MB} & \% FPR & 75.9827 & 78.6056 & 81.2512 & 83.3572 & 85.0294 \\ \cline{2-7}
& \% FNR & 9.2090 & 7.8180 & 6.4065 & 5.3588 & 4.5864\\ \hline \hline
\multirow{2}{*}{128 MB} & \% FPR & 21.0876 & 15.5928 & 15.2615 & 16.4467 & 18.2903 \\ \cline{2-7}
& \% FNR & 9.8405 & 17.7438 & 22.9224 & 26.2752 & 28.3608 \\ \hline \hline
\multirow{2}{*}{512 MB} & \% FPR & 6.4618 & 1.9692 & 1.0095 & 0.6880 & 0.5590 \\ \cline{2-7}
& \% FNR & 3.3399 & 6.8855 & 9.9011 & 12.6723 & 15.1324 \\ \hline \hline
\end{tabular}
\caption{Synthetic Dataset of 1B elements (60\% distinct)}
\label{BSBFSD:1B:60distinct}
\end{center}
\end{table}

Finally, we present the performance of RLBSBF algorithm across various parametric setting of $k$. We use the same random dataset of size 1B with 60\% distinct element and again vary the memory size from 8MB to 512MB. 
From Table \ref{RLBSBF:1B:60distinct}, we observe that RLBSBF exhibits similar performance as above. Likewise we set $k=2$.

\begin{table}[htbp]
\begin{center}
\begin{tabular}{|l|l|l|l|l|l|l|}
\hline
\multicolumn{7}{|c|}{Algorithm:RLBSBF, Dataset:1B , Distinct:60\% } \\ \hline
Space &  & k=1 & k=2 & k=3 & k=4 & k=5 \\ \hline \hline
\multirow{2}{*}{8 MB} & \% FPR & 81.4726 & 52.9291 & 30.2439 & 17.3562 & 10.1282 \\ \cline{2-7}
& \% FNR & 5.3103 & 40.466 & 66.0585 & 80.3262 & 88.3195\\ \hline \hline
\multirow{2}{*}{128 MB} & \% FPR & 22.8676 & 16.3554 & 12.3957 & 8.56757 & 5.5645 \\ \cline{2-7}
& \% FNR & 2.4337 & 15.0161 & 35.6592 & 55.2835 & 69.3118 \\ \hline \hline
\multirow{2}{*}{512 MB} & \% FPR & 6.6563 & 2.07884 & 1.0555 & 0.6773 & 0.4828 \\ \cline{2-7}
& \% FNR & 0.2553 & 1.98971 & 6.1959 & 13.3911 & 23.2844 \\ \hline \hline
\end{tabular}
\caption{Synthetic Dataset of 1B elements (60\% distinct)}
\label{RLBSBF:1B:60distinct}
\end{center}
\end{table}

\subsection{Quality Comparison}
\label{subsec:qual}

In this section we present the variation of FPR and FNR along with convergence to stability for SBF, RSBF, BSBF, BSBFSD and RLBSBF with increasing number of records in the input stream. 
%
%
We present graph based analysis of FPR and FNR performances of our algorithms in comparison with SBF for 1B data with 15\% distinct element. Fig.~\ref{fig:fpr_fnr128} shows the graph for FPR and FNR when 
128MB memory space is available. We observe that all the algorithms achieve quite low FPR in the range of 1\%-2\%. Also the FPR becomes stable at around 300M-350M data points for almost all the 
variations. However we observe a sharp contrast in FNR level. While SBF exhibits nearly 45\% FNR, RLBSBF exhibits less than 1\% FNR which is almost 70x times improvement. Other variations also exhibit 
strong improvements in the FNR performance in comparison with SBF. We also observe that for BSBF, BSBFSD and RLBSBF the FNR level keeps on decreasing as the stream length increases. This validates our 
theoretical claim about low FNR for these algorithms.
\begin{figure}[htbp]
\begin{center}
	\includegraphics[width=0.9\columnwidth]{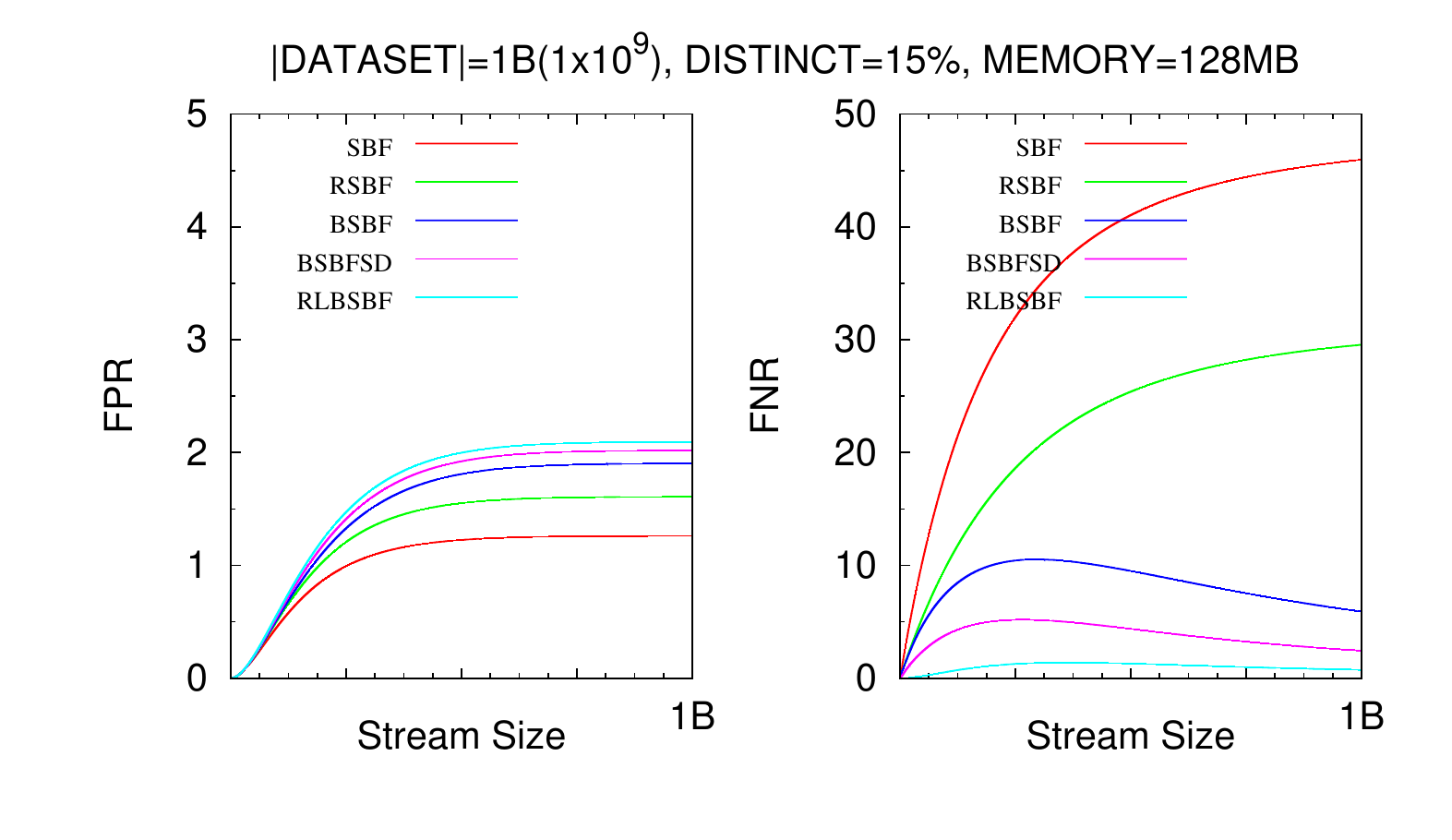}
	\caption{FPR and FNR Performances.}
	\label{fig:fpr_fnr128}
\end{center}
\end{figure}
%
%
Fig.~\ref{fig:fpr_fnr256} presents the FPR and FNR for 1B random dataset with 256MB memory space. FPR for all the algorithms falls within a small range of 0.4\%-0.6\%, however FNR varies from 30\% for SBF to 
0.2\% for RLBSBF. Also the curves provide empirical evidence for our theoretical results that FNR tends to zero as we increase the stream length. SBF does not exhibit this property and its FNR increases 
as the stream length increases.
%
\begin{figure}[htbp]
\begin{center}
	\includegraphics[width=0.9\columnwidth]{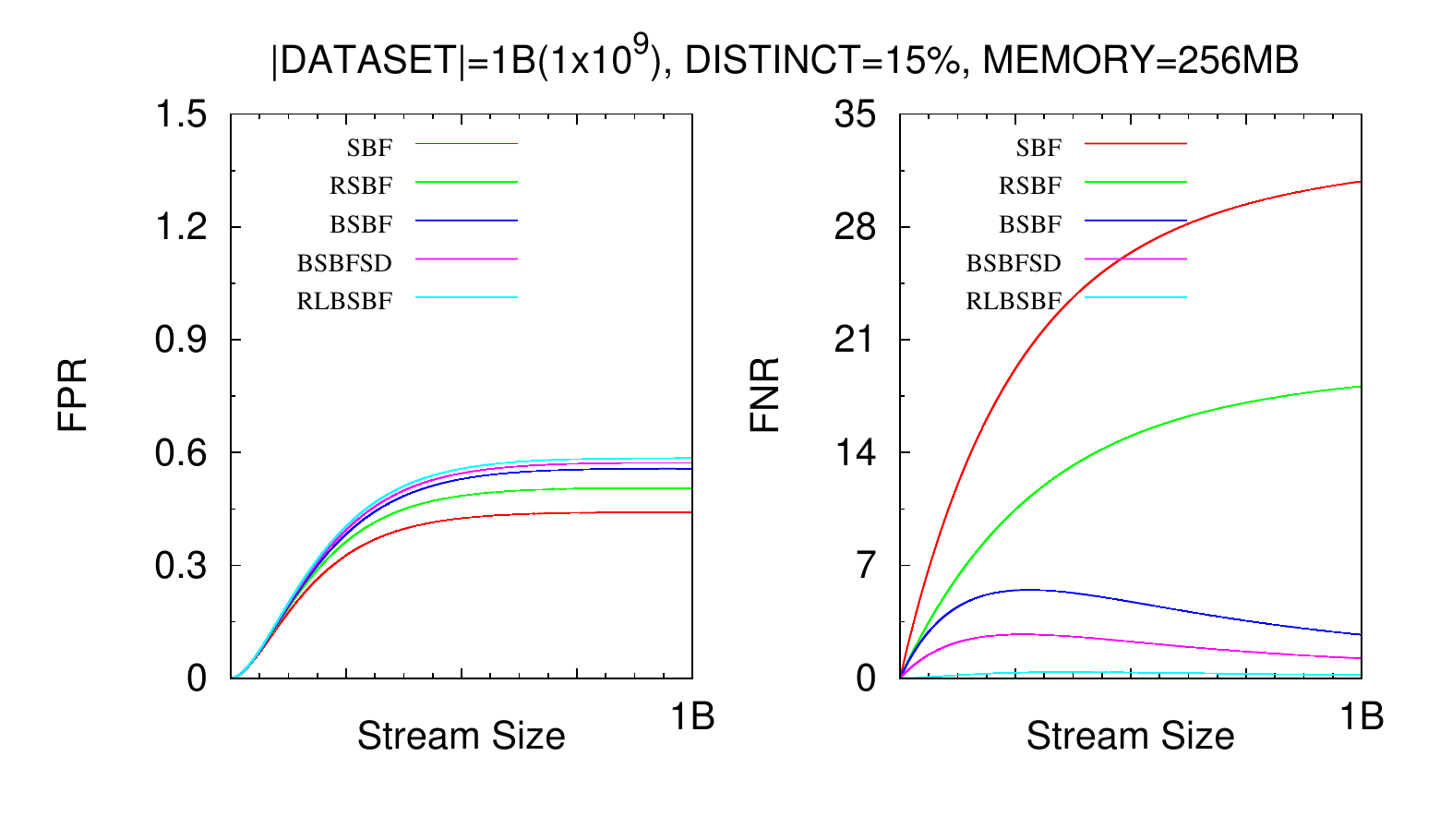}
	\caption{FPR and FNR Performances.}
	\label{fig:fpr_fnr256}
\end{center}
\end{figure}
%
%
Fig.~\ref{fig:fpr_fnr512} shows that as we increase memory space from 256MB to 512MB, FNR further drops to a very low limit (less than 1\%) for BSBF, BSBFSD, RLBSBF. RSBF also achieves an improvement of 2x in 
terms of FNR over SBF. FPR of all the algorithm stabilizes at a very low level and remains comparable.
%
\begin{figure}[htbp]
\begin{center}
	\includegraphics[width=0.9\columnwidth]{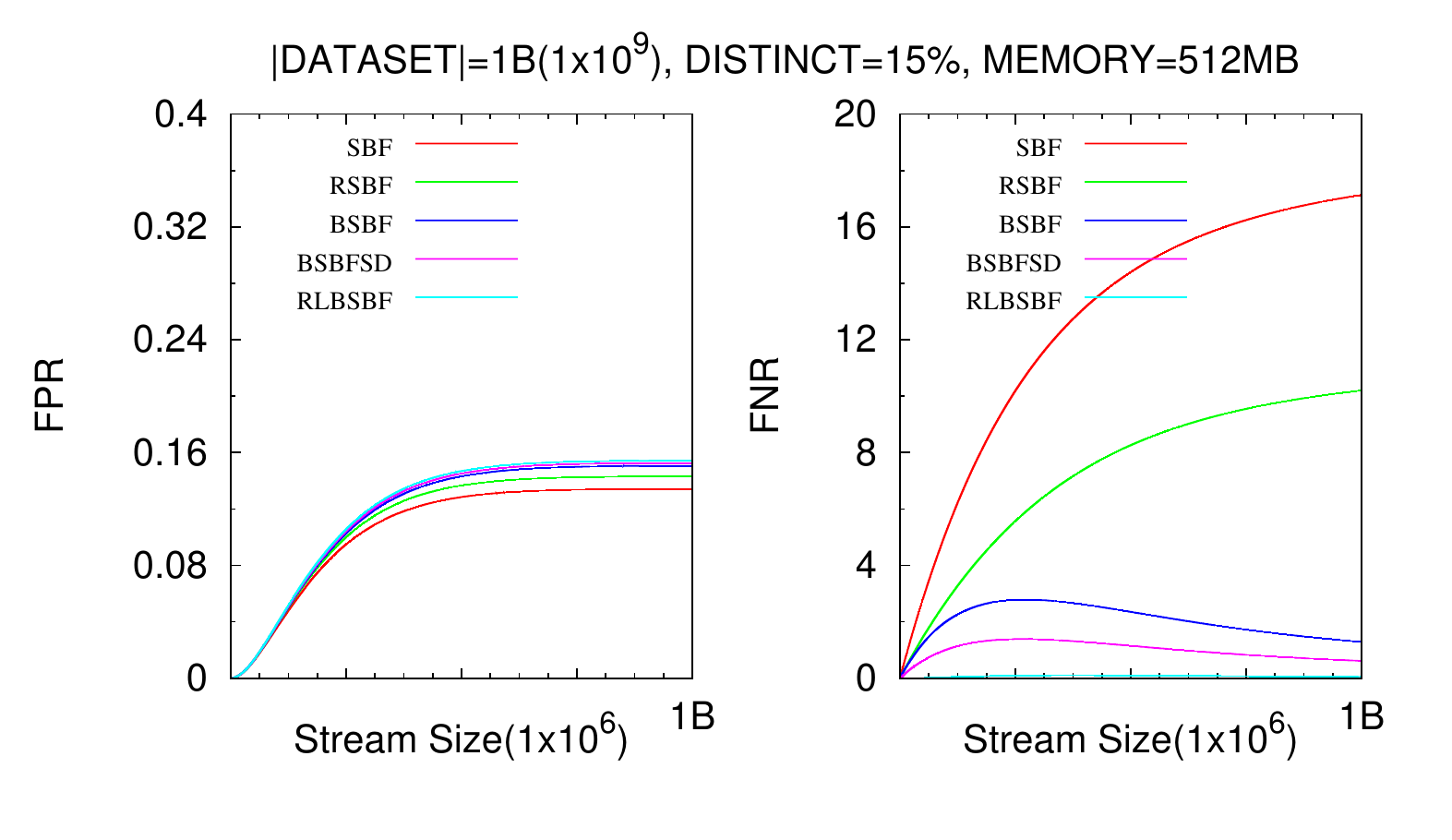}
	\caption{FPR and FNR Performances.}
	\label{fig:fpr_fnr512}
\end{center}
\end{figure}

%
%

We next compare the FPR and FNR of the proposed algorithms for the 1B synthetic random dataset but with 60\% distinct element. We vary the memory size from 128MB to 512MB. To emphasize the comparability of FPR and 
large improvements of FNR between SBF and our algorithms, we have plotted the graphs using natural scale. Fig.~\ref{fig:fpr_fnr12860} shows FPR and FNR comparison for 128MB memory space. Although SBF exhibits good FPR, 
it has a very poor FNR(70\%). RLBSBF curbs down the FNR to almost 10\% while losing merely 3\% on FPR. As evident from the graph, the gain over FNR dominates the loss over FPR. Also BSBF and BSBFSD attain stability 
in FNR ratio from around 500M data points.
%
%
\begin{figure}[htbp]
\begin{center}
	\includegraphics[width=0.9\columnwidth]{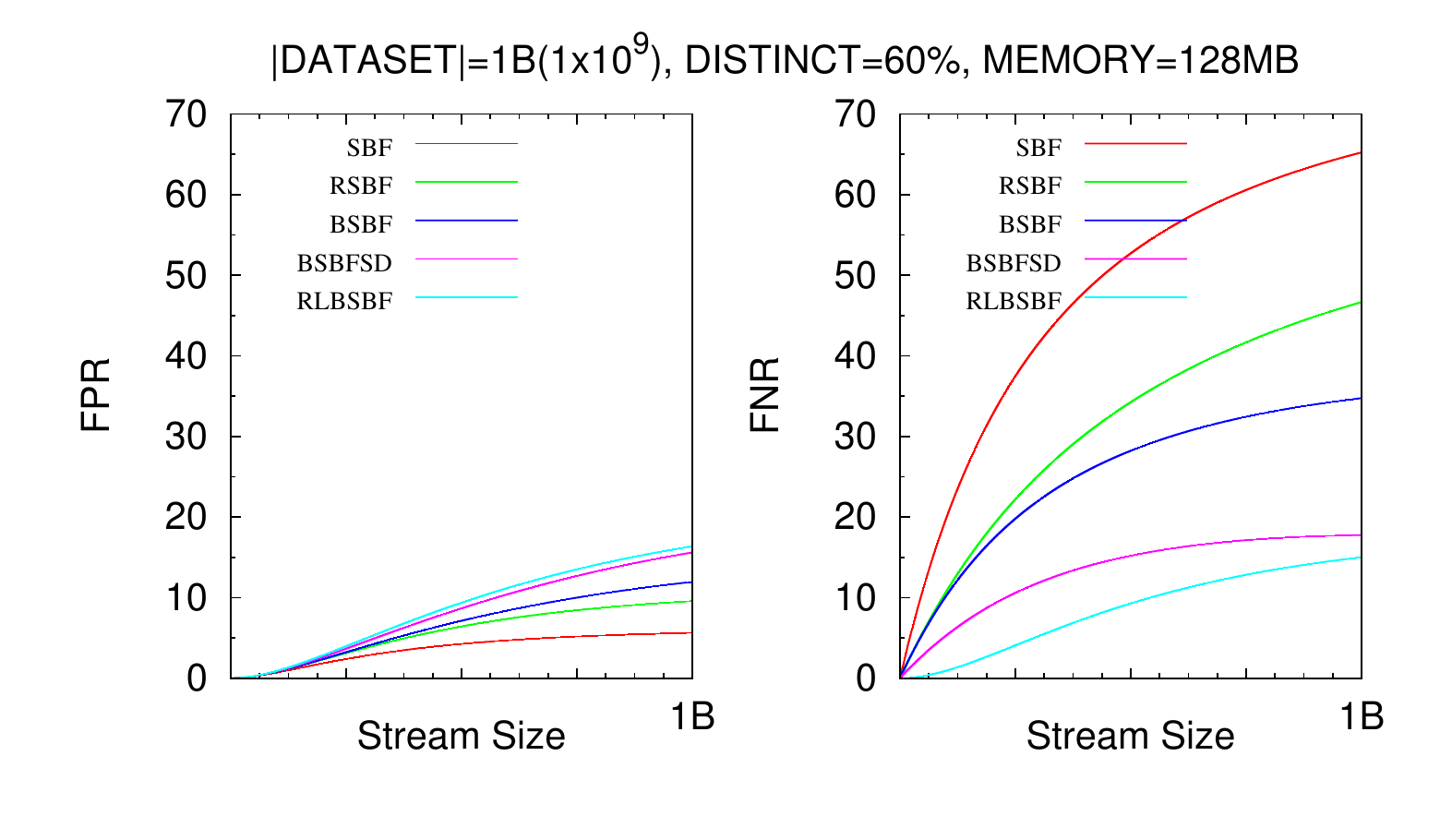}
	\caption{FPR and FNR Performances.}
	\label{fig:fpr_fnr12860}
\end{center}
\end{figure}
%
%
%
Fig.~\ref{fig:fpr_fnr25660} depicts the FPR and FNR of various algorithm for 1B data(60\% distinct) when 256MB memory space is used. We observe that the gap in FPR among the algorithms have been reduced significantly. 
The range of FPR for the algorithms lies between 4\%-7\%. But the difference in FNR achieved by the algorithms is very high. RSBF, BSBF, BSBFSD and RLBSBF provides an improvement of 2x, 2.5x, 5x and 10x respectively. 
We also observe a similar FNR stability trend as previously for 128MB memory.
%
%
\begin{figure}[htbp]
\begin{center}
	\includegraphics[width=0.9\columnwidth]{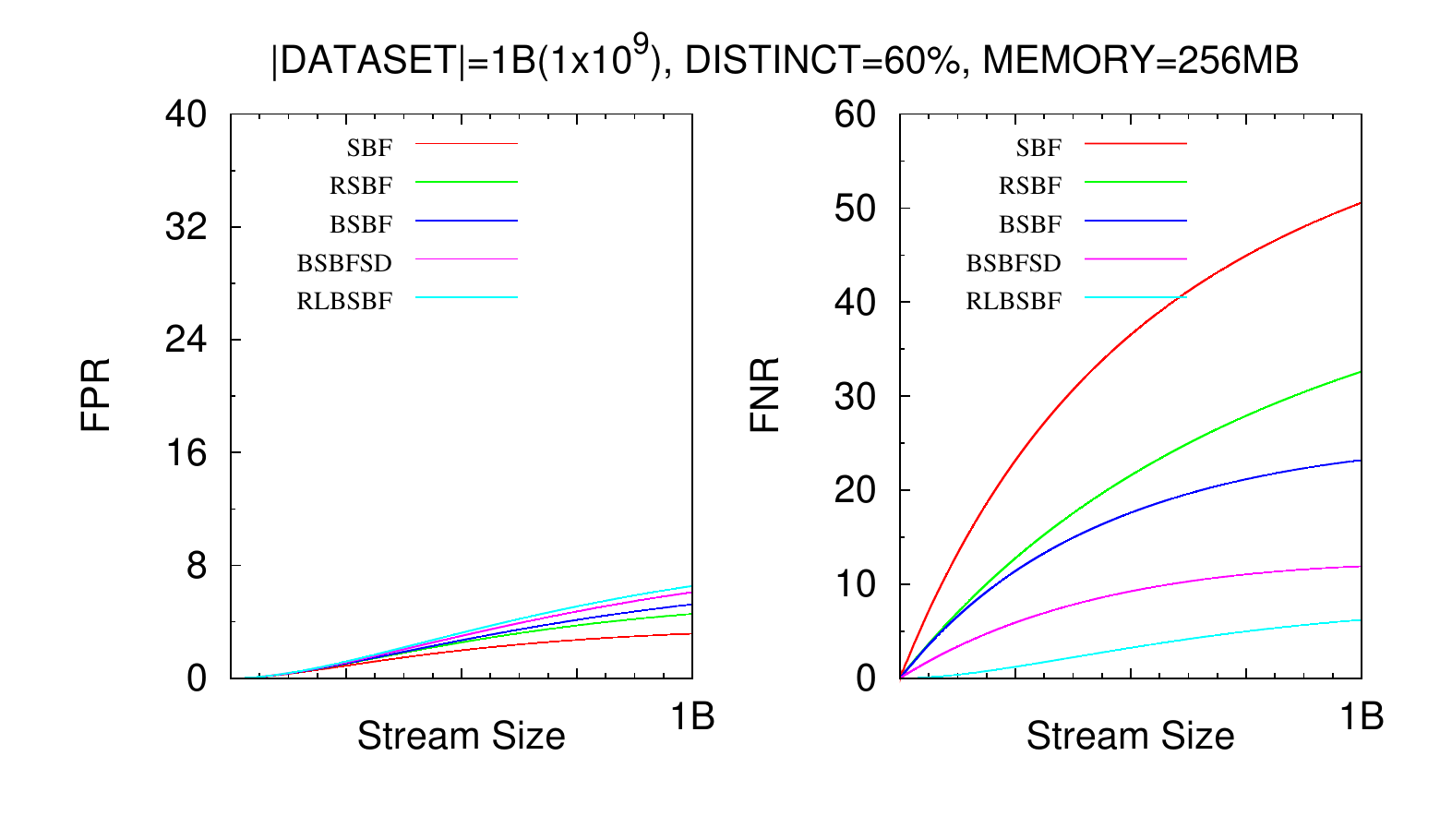}
	\caption{FPR and FNR Performances.}
	\label{fig:fpr_fnr25660}
\end{center}
\end{figure}
%
%
%
Finally Fig.~\ref{fig:fpr_fnr51260} shows that the FPR of the algorithms converges to a very low limit when 512MB memory space is used. But the difference in FNR level only keeps on improving in favor of our proposed 
algorithms. This also shows the scalability aspect of our algorithms. As we increase the memory space, RSBF, BSBF, BSBFSD and RLBSBF improves in performance much faster than SBF.
%
%
\begin{figure}[htbp]
\begin{center}
	\includegraphics[width=0.9\columnwidth]{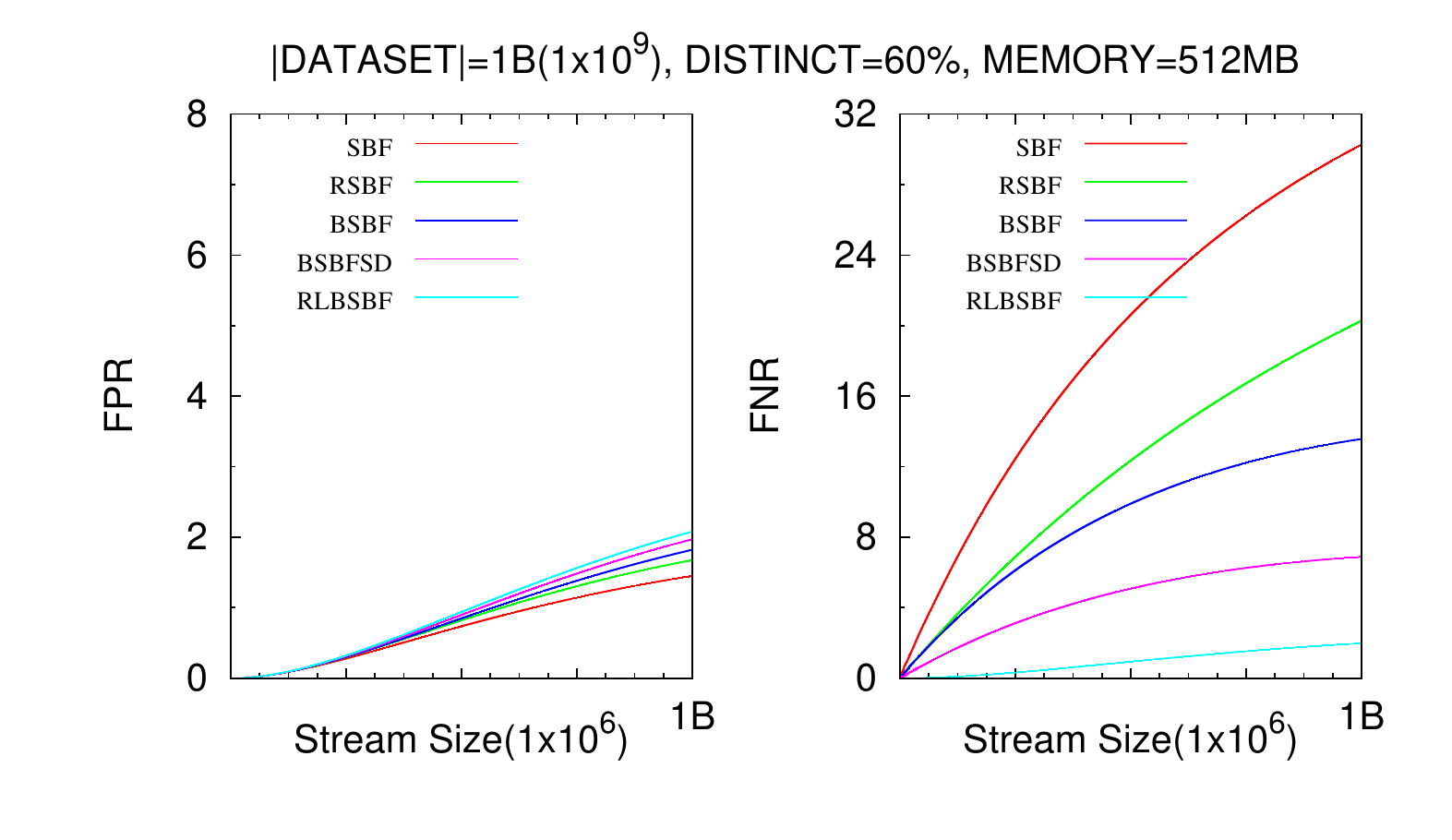}
	\caption{FPR and FNR Performances.}
	\label{fig:fpr_fnr51260}
\end{center}
\end{figure}

\begin{figure}[htbp]
\begin{center}
	\includegraphics[width=0.9\columnwidth]{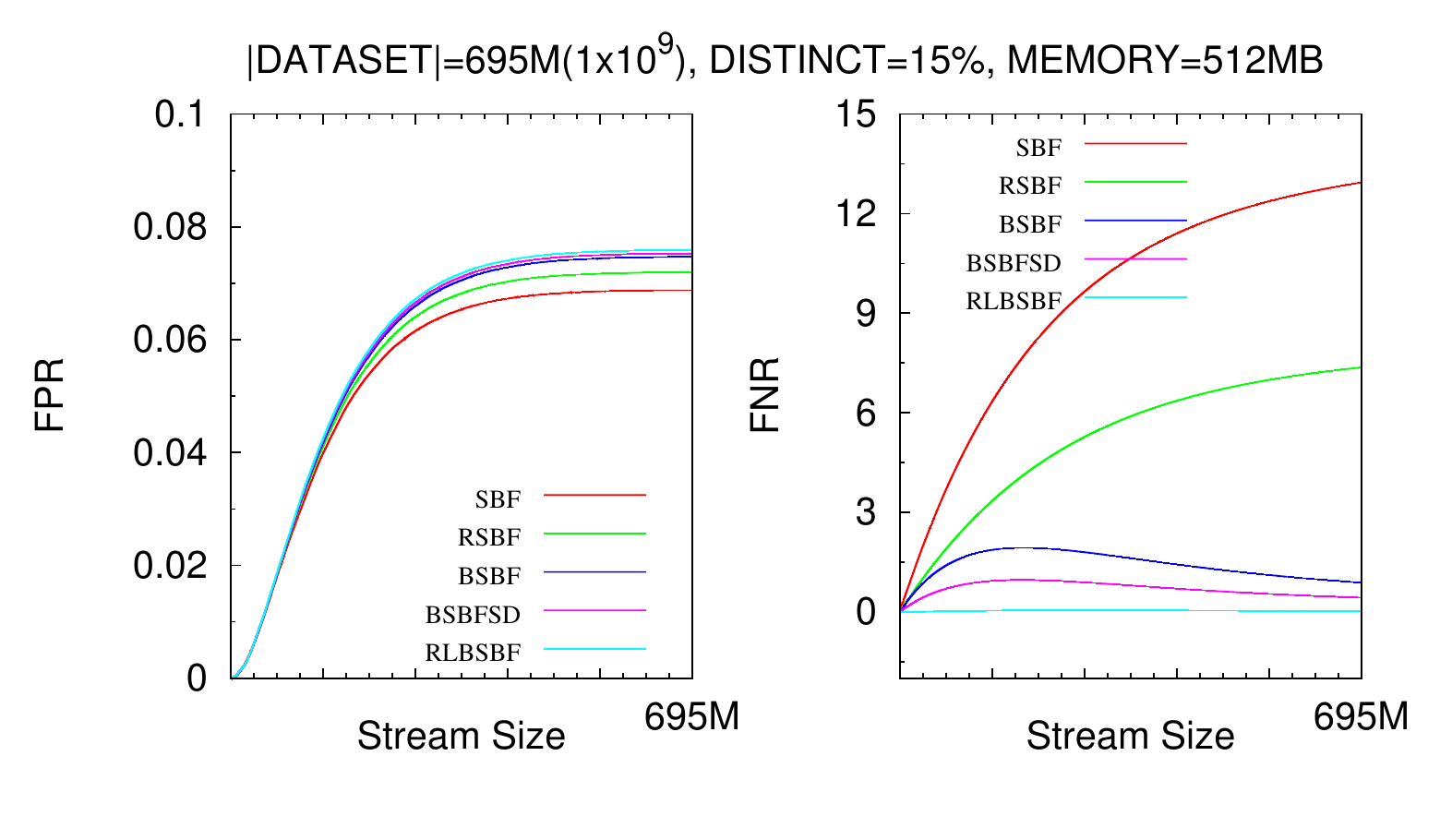}
	\caption{FPR and FNR Performances.}
	\label{fig:fpr_fnr15512}
\end{center}
\end{figure}
\begin{figure}[htbp]
\begin{center}
	\includegraphics[width=0.9\columnwidth]{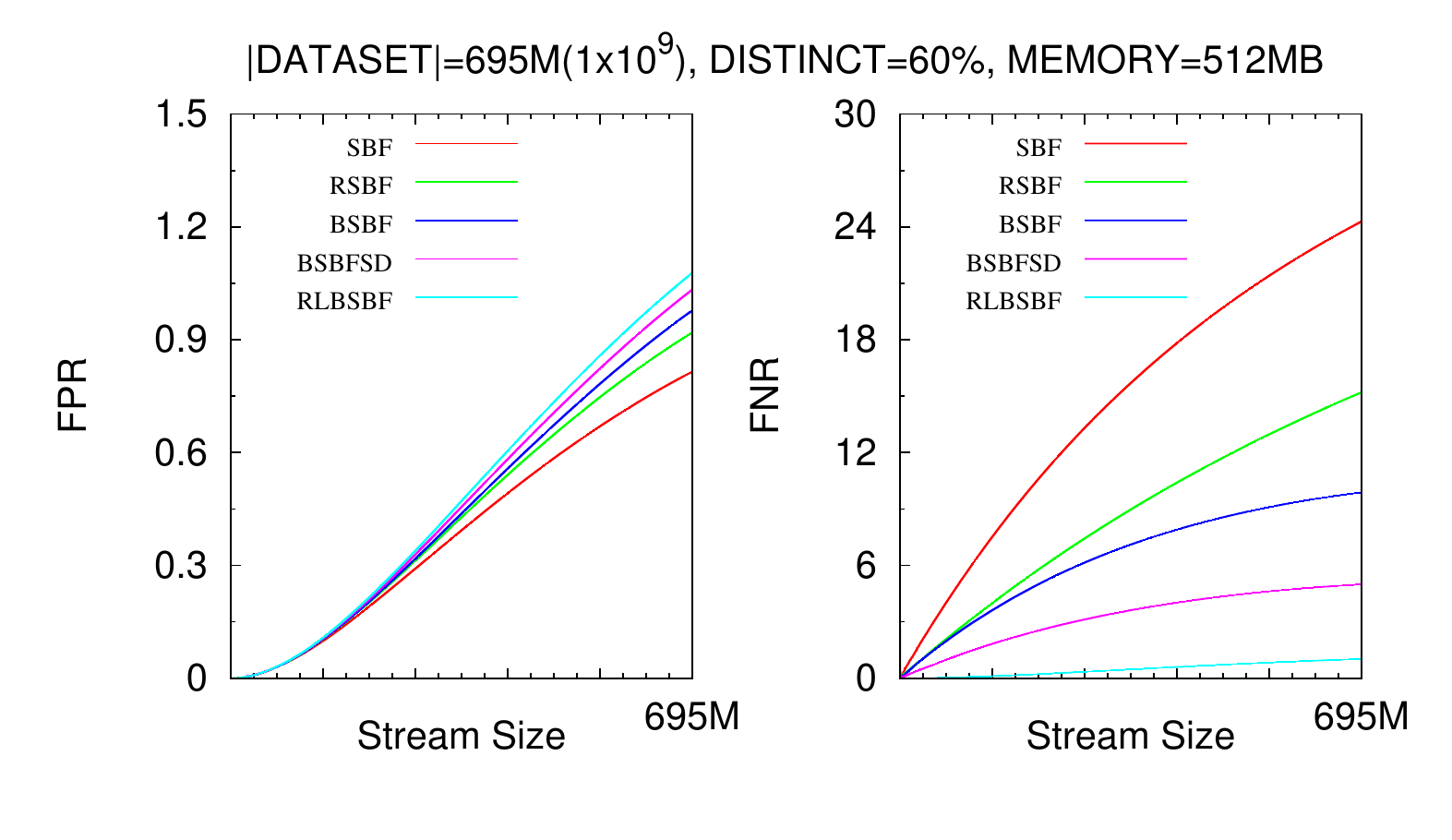}
	\caption{FPR and FNR Performances.}
	\label{fig:fpr_fnr60512}
\end{center}
\end{figure}
\begin{figure}[htbp]
\begin{center}
	\includegraphics[width=0.9\columnwidth]{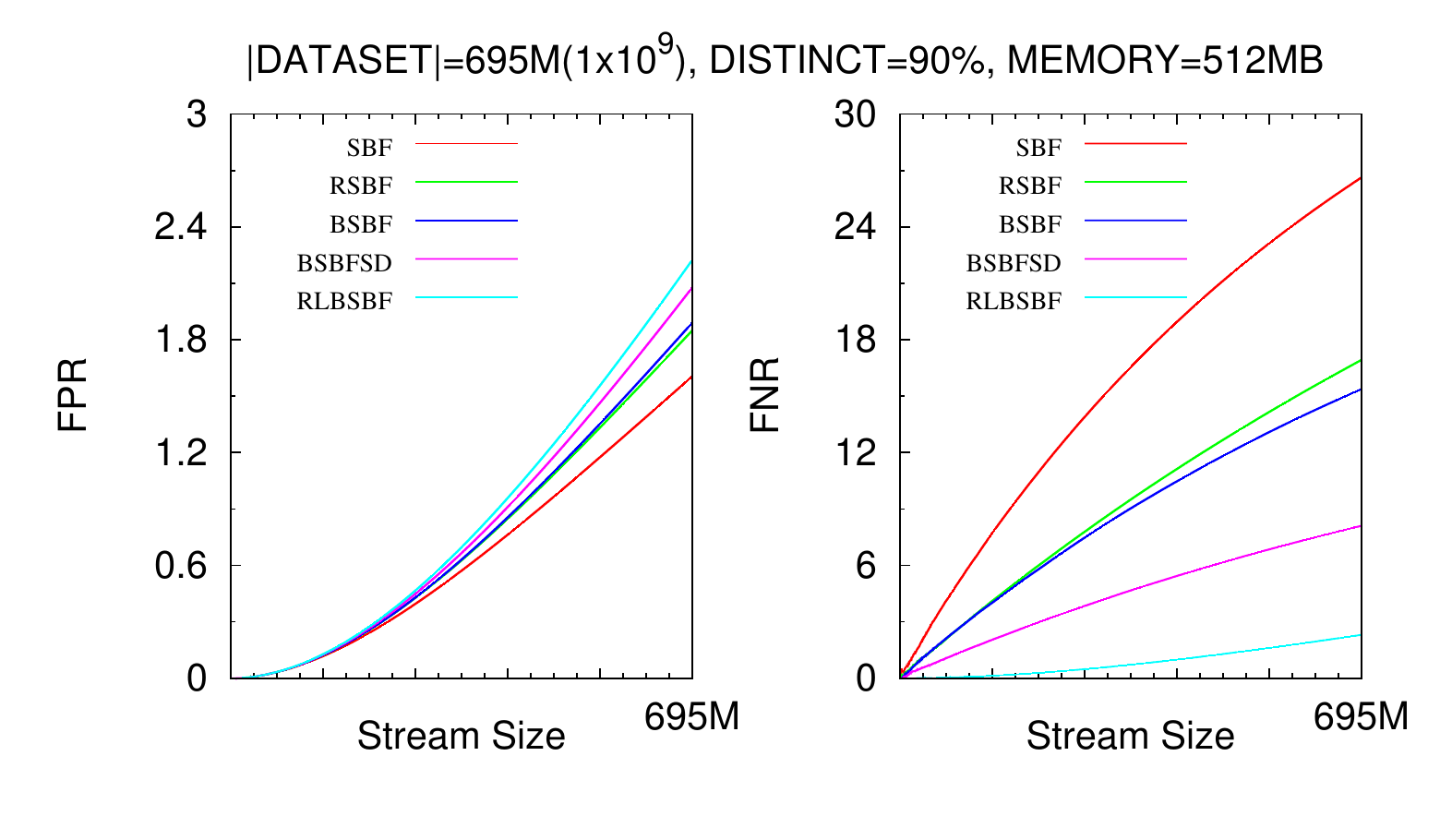}
	\caption{FPR and FNR Performances.}
	\label{fig:fpr_fnr90512}
\end{center}
\end{figure}
%
Figs.~\ref{fig:fpr_fnr15512},~\ref{fig:fpr_fnr60512} and~\ref{fig:fpr_fnr90512} compare the FPR and FNR of the various algorithms with increasing distinct percentage in data stream keeping the memory space fixed at 512MB. We generate 
three uniform random dataset of 695M elements with 15\%, 60\% and 90\% distinct element. We observe that the proposed algorithms outperforms SBF in FNR and convergence rates with comparable FPR.

Fig.~\ref{fig:fpr_fnr15512} shows FPR for all the algorithm is almost similar and stable for 695M dataset with 15\% distinct element. For FNR, our algorithms completely outperform SBF. Fig.~\ref{fig:fpr_fnr60512} 
represents the FPR and FNR scenario for 60\% distinct data stream of size 695M with 512MB size. We observe that at 695M data point, none of the algorithms have achieved stability in terms of FPR. In terms of FNR 
only BSBF, BSBFSD and RLBSBF are inching towards stability while SBF is far way off from stability.  Fig.~\ref{fig:fpr_fnr90512} shows the situation for 90\% distinct element with similar trends as discussed 
above. Hence our proposed algorithms always outperform SBF comprehensively in terms of FNR and also attain comparable of FPR and stability.



We next exhibit the stability of our algorithms. We compute the load of the Bloom Filter structure at every point of the stream. We define load as the number of $1$'s in the Bloom Filters normalized 
by the total memory space in bits. We present the load graph for various algorithms for 1B data with 15\% distinct.  Fig.~\ref{fig:ones} shows the load graph when 256MB memory space is allocated 
and when 512MB memory is used. It can be observed from the graph that all the algorithms nearly attain stability when 300M-400M data points have been 
processed. Also, as we increase the memory space from 256MB to 512MB, the stability is reached faster. The convergence rate of the algorithms are similar to each other.
%
\begin{figure}[htbp]
\begin{center}
	\includegraphics[width=0.9\columnwidth]{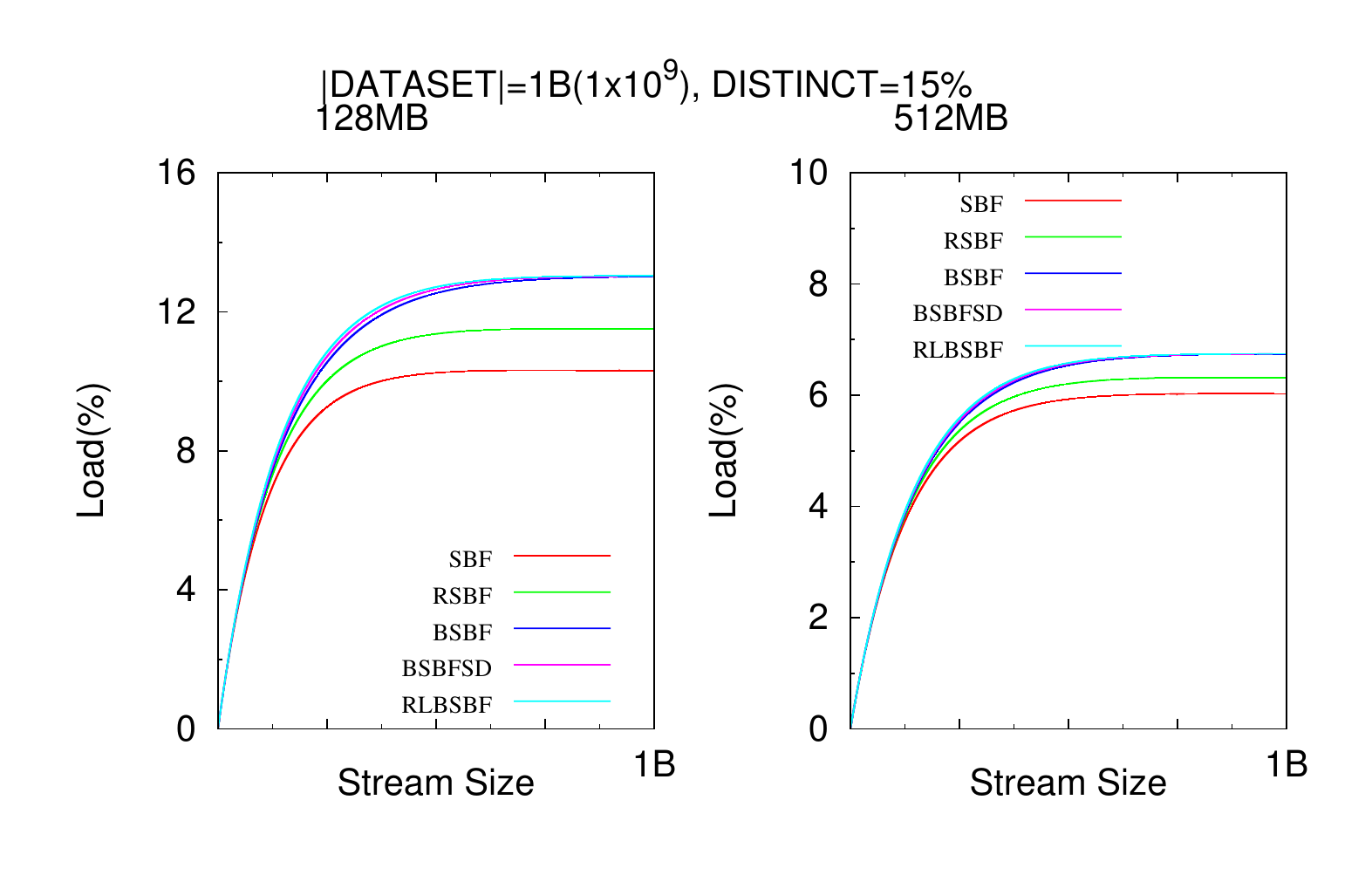}
	\caption{Stability Performance.}
	\label{fig:ones}
\end{center}
\end{figure}

\subsection{Detailed Analysis}
\label{subsec:detail}

In this section, we present detailed experimental analysis of the various algorithms that we have designed in the previous sections. We compare the performance of various algorithms against variation 
of memory used and percentage of distinct elements in the stream.
%
In the following comparisons we have set $k=2$ for RSBF, BSBF, BSBFSD, RLBSBF algorithms.

Table~\ref{695M:15distinct} presents the FPR and FNR with 695M records with 15\% distinct element in the stream. The memory size varies from 64MB to 512MB for the underlying Bloom Filter based data structures. 
Here we observe that although SBF performs reasonably well in terms of FPR for low memory (64MB), FNR degrades to an unacceptable limit ($53.26\%$). For 64MB memory, RLBSBF keeps the FPR and FNR to very low 
limits (FPR of 3.7\% and FNR of 1.3\%). But for 128MB memory all the five algorithms exhibits comparable FPR with lowest being SBF (0.74\%) and highest being RLBSBF(1.08\%). But FNR for BSBF, BSBFSD and RLBSBF 
improve drastically compared to that of SBF. For example, RLBSBF gains an improvement of 100x times over SBF in terms of FNR! This behavior continues for higher memory size and for 512MB memory, all the four 
proposed algorithms outperform SBF comprehensively in term of FNR while keeping the FPR level competitive. For example RLBSBF attains an improvement of 600x in respect to FNR when compared to SBF while FPR level 
for both remains almost the same.

\begin{table}[htbp]
\begin{center}
\begin{tabular}{|l|l|l|l|l|l|l|}
\hline
\multicolumn{7}{|c|}{Dataset:695M , Distinct:15\% } \\ \hline
Space &  & SBF & RSBF & BSBF & BSBFSD & RLBSBF   \\ \hline \hline
\multirow{2}{*}{64 MB} & \% FPR & 1.9319 & 2.6276 & 3.2569 & 3.5475 & 3.7064 \\ \cline{2-7}
& \% FNR & 53.2681 & 35.9014 & 8.7547 & 3.3299 & 1.3453 \\ \hline \hline
\multirow{2}{*}{128 MB} & \% FPR & 0.7414 & 0.8903 & 1.0128 & 1.0530 & 1.0821 \\ \cline{2-7}
& \% FNR & 37.7853 & 23.126 & 3.888 & 1.7048 & 0.3773 \\ \hline \hline
\multirow{2}{*}{256 MB} & \% FPR & 0.2384 & 0.2637 & 0.2822 & 0.2881 & 0.2924 \\ \cline{2-7}
& \% FNR & 23.8930 & 13.5047 & 1.8199 & 0.8551 & 0.1010 \\ \hline \hline
\multirow{2}{*}{512 MB} & \% FPR & 0.0688 & 0.0719 & 0.0747 & 0.0753 & 0.0759\\ \cline{2-7}
& \% FNR & 12.9392 & 7.3674 & 0.8794 & 0.4267 & 0.0262 \\ \hline \hline
\end{tabular}
\caption{Synthetic Dataset of 695M elements (15\% distinct)}
\label{695M:15distinct}
\end{center}
\end{table}

Table~\ref{695M:60distinct} presents the FPR and FNR with 695M records with 60\% distinct element in the stream. The memory size varies from 64MB to 512MB. For 8MB memory, as before SBF outperforms 
other algorithm in terms of FPR. But an FNR of 70.832\% makes SBF unsuitable for all practical purposes. As we increase the memory size to 128MB, 256MB and 512MB, we observe sharp fall in the FPR level 
of all the algorithms. For 512MB, all the algorithms offers FPR that is in the range of 0.8\%-1.07\%. But in terms of FNR, our proposed algorithms outperforms SBF significantly. Notably, RLBSBF achieves 
an improvement of around 24 times in terms of FNR compared to SBF. We also observe that with increase in memory space, FNR of RLBSBF falls sharply to a very low limit. For example, if we double the memory 
space from 128MB to 256MB and then 256MB to 512MB, FNR of RLBSBF drops 3 times in each occasion, hence performing the best among the algorithms in terms of FNR.

\begin{table}[htbp]
\begin{center}
\begin{tabular}{|l|l|l|l|l|l|l|}
\hline
\multicolumn{7}{|c|}{Dataset:695M , Distinct:60\% } \\ \hline
Space &  & SBF & RSBF & BSBF & BSBFSD & RLBSBF   \\ \hline \hline
\multirow{2}{*}{64 MB} & \% FPR & 6.7672 & 12.2408 & 16.0723 & 22.4491 & 22.7214 \\ \cline{2-7}
& \% FNR & 70.832 & 52.8855 & 39.9809 & 19.9394 & 20.2164 \\ \hline \hline
\multirow{2}{*}{128 MB} & \% FPR & 4.3075 & 6.7265 & 8.0154 & 9.7951 & 10.4734\\ \cline{2-7}
& \% FNR & 57.9713 & 39.2682 & 28.6708 & 14.7555 & 9.8051 \\ \hline \hline
\multirow{2}{*}{256 MB} & \% FPR & 2.0861 & 2.7654 & 3.0816 & 3.4335 & 3.6588\\ \cline{2-7}
& \% FNR & 42.0788 & 25.7648 & 17.7509 & 9.0574 & 3.4843 \\ \hline \hline
\multirow{2}{*}{512 MB} & \% FPR & 0.8151 & 0.9189 & 0.9779 & 1.0331 & 1.0775 \\ \cline{2-7}
& \% FNR & 24.2902 & 15.2035 & 9.8819 & 4.9892 & 1.0236\\ \hline \hline
\end{tabular}
\caption{Synthetic Dataset of 695M elements (60\% distinct)}
\label{695M:60distinct}
\end{center}
\end{table}

We present the FPR and FNR of 695M records with 90\% distinct element in Table~\ref{695M:90distinct}. The memory size again varies from 64MB to 512MB. For various memory size, the algorithms 
exhibits similar trends as before with all the proposed variations providing large improvements in terms of FNR over SBF while keeping the FPR at comparable level. For example at the cost of 
1.5x times worse FPR (from 1.6\% to 2.2\%), RLBSBF attains an improvement of 13x time in FNR over SBF.

\begin{table}[htbp]
\begin{center}
\begin{tabular}{|l|l|l|l|l|l|l|}
\hline
\multicolumn{7}{|c|}{Dataset:695M , Distinct:90\% } \\ \hline
Space &  & SBF & RSBF & BSBF & BSBFSD & RLBSBF   \\ \hline \hline
\multirow{2}{*}{64 MB} & \% FPR & 8.5028 & 16.3502 & 20.6486 & 31.104 & 30.1024 \\ \cline{2-7}
& \% FNR & 73.134 & 55.0285 & 45.7102 & 24.1625 & 26.8786 \\ \hline \hline
\multirow{2}{*}{128 MB} & \% FPR & 6.3252 & 10.6716 & 11.9222 & 15.7484 & 16.6894 \\ \cline{2-7}
& \% FNR & 61.1759 & 41.9224 & 36.5123 & 20.3118 & 16.4363 \\ \hline \hline
\multirow{2}{*}{256 MB} & \% FPR & 3.5908 & 5.0678 & 5.3168 & 6.3016 & 6.8635 \\ \cline{2-7}
& \% FNR & 45.4987 & 28.179 & 25.3195 & 13.8078 & 7.0784 \\ \hline \hline
\multirow{2}{*}{512 MB} & \% FPR & 1.6058 & 1.8497 & 1.8911 & 2.0783 & 2.2248 \\ \cline{2-7}
& \% FNR & 26.6313 & 16.9302 & 15.3766 & 8.1061 & 2.3059 \\ \hline \hline
\end{tabular}
\caption{Synthetic Dataset of 695M elements (90\% distinct)}
\label{695M:90distinct}
\end{center}
\end{table}

Table~\ref{1B:15distinct} presents the FPR and FNR with 1B records with 15\% distinct element in the stream. The memory size varies from 64MB to 512MB for the underlying Bloom Filter based data structures. 
Here we observe that SBF performs very poorly in terms of FNR for low memory size. For 64MB memory, RLBSBF keeps the FPR and FNR to very low limits and achieves almost 30x improvement over SBF in FNR while 
losing only 2.5x in FPR. As the memory size increases, FPR for all our algorithms becomes stable at a very low thresholds ($0.1\%-0.2\%$) similar to that of SBF. But the gain in FNR performance is enormous. 
While SBF has a FNR of $17.1336\%$ for 512MB memory, RSBF attains $10.2015\%$ FNR, nearly a 2x improvement. BSBF, BSBFSD, RLBSBF observes an improvement of 17x times, 27x times and 317x times respectively. 

\begin{table}[htbp]
\begin{center}
\begin{tabular}{|l|l|l|l|l|l|l|}
\hline
\multicolumn{7}{|c|}{Dataset:1B , Distinct:15\% } \\ \hline
Space &  & SBF & RSBF & BSBF & BSBFSD & RLBSBF \\ \hline \hline
\multirow{2}{*}{64 MB} & \% FPR & 2.9156 & 4.2891 & 5.5775 & 6.3441 & 6.6755 \\ \cline{2-7}
& \% FNR & 60.8981 & 43.1705 & 13.7096 & 4.6175 & 2.5795 \\ \hline \hline
\multirow{2}{*}{128 MB} & \% FPR & 1.2608 & 1.6079 & 1.9023 & 2.0181 & 2.0930 \\ \cline{2-7}
& \% FNR & 45.9566 & 29.5540 & 5.9096 & 2.4357 & 0.7400 \\ \hline \hline
\multirow{2}{*}{256 MB} & \% FPR & 0.4413 & 0.5059 & 0.5572 & 0.5727 & 0.5849 \\ \cline{2-7}
& \% FNR & 30.8269 & 18.1142 & 2.6956 & 1.2296 & 0.2026 \\ \hline \hline
\multirow{2}{*}{512 MB} & \% FPR & 0.1341 & 0.1431 & 0.1506 & 0.1526 & 0.1543\\ \cline{2-7}
& \% FNR & 17.1336 & 10.2015  & 1.2846 & 0.6139 & 0.0535 \\ \hline \hline
\end{tabular}
\caption{Synthetic Dataset of 1B elements (15\% distinct)}
\label{1B:15distinct}
\end{center}
\end{table}

Table~\ref{1B:60distinct} presents the FPR and FNR with 1B records with 60\% distinct element in the stream. The memory size varies, as before, from 64MB to 512MB for the underlying Bloom Filter based 
data structures. Here also we observe FNR for SBF degrades to an unacceptable limit for 64MB or 128MB memory size. For 64MB memory, RLBSBF keeps the FPR and FNR to reasonable limits. But as the memory 
size increases, FPR for all our algorithms becomes stable at a very low point ($1\%-2\%$) similar to that of SBF. But FNR performance is improved drastically for all our algorithms. While SBF has a FNR 
of $30.2739\%$ for 512MB memory, RSBF attains an FNR of $20.2770\%$, nearly a 1.5 times improvement. BSBF, BSBFSD, RLBSBF observes an improvement of 2.5 times, 5 times and 30 times respectively. This 
improvements come at the cost of 2 time more FPR in the worst case (RLBSBF).

\begin{table}[htbp]
\begin{center}
\begin{tabular}{|l|l|l|l|l|l|l|}
\hline
\multicolumn{7}{|c|}{Dataset:1B , Distinct:60\% } \\ \hline
Space &  & SBF & RSBF & BSBF & BSBFSD & RLBSBF   \\ \hline \hline
\multirow{2}{*}{64 MB} & \% FPR & 7.8507 & 15.0561 & 20.7283 & 31.3858 & 30.044 \\ \cline{2-7}
& \% FNR & 75.7271 & 58.6734 & 44.9042 & 21.0835 & 25.7526\\ \hline \hline
\multirow{2}{*}{128 MB} & \% FPR & 5.6378 & 9.5676 & 11.9472 & 15.5928 & 16.3554\\ \cline{2-7}
& \% FNR & 65.2456 & 46.6885 & 34.7514 & 17.7438 & 15.0161\\ \hline \hline
\multirow{2}{*}{256 MB} & \% FPR & 3.1629 & 4.5557 & 5.2369 & 6.0937 & 6.5378\\ \cline{2-7}
& \% FNR & 50.5726 & 32.6139 & 23.2001 & 11.9067 & 6.2079\\ \hline \hline
\multirow{2}{*}{512 MB} & \% FPR & 1.4504 & 1.6747 & 1.82011 & 1.9692 & 2.0788\\ \cline{2-7}
& \% FNR & 30.2739 & 20.2770 & 13.5658 & 6.8855 & 1.9897\\ \hline \hline
\end{tabular}
\caption{Synthetic Dataset of 1B elements (60\% distinct)}
\label{1B:60distinct}
\end{center}
\end{table}

Table~\ref{1B:90distinct} presents the FPR and FNR with 1B records with 90\% distinct element in the stream. The memory size varies, as before, from 64MB to 512MB. Here also we observe that at low memory 
only RLBSBF has both FPR and FNR at around 50\%. All other algorithm exhibits either very poor FPR or very poor FNR. However as the memory size increases, FPR for all our algorithms becomes stable at $2\%-4\%$ 
similar to that of SBF. While SBF attains a FPR of 2.72\%, RSBF, BSBF, BSBFSD and RLBSBF also attains a FPR that is comparable with that of SBF. But the gain in FNR performance is huge. While SBF has a FNR of 
$32.8335\%$ for 512MB memory, RSBF has $22\%$ FNR, nearly a 1.5 times improvement. BSBF, BSBFSD, RLBSBF observes an improvement of 1.5 times, 3 times and 8 times respectively. 

\begin{table}[htbp]
\begin{center}
\begin{tabular}{|l|l|l|l|l|l|l|}
\hline
\multicolumn{7}{|c|}{Dataset:1B , Distinct:90\% } \\ \hline
Space &  & SBF & RSBF & BSBF & BSBFSD & RLBSBF   \\ \hline \hline
\multirow{2}{*}{64 MB} & \% FPR & 9.2617 & 18.6342 & 25.0083 & 40.5501 & 36.7959\\ \cline{2-7}
& \% FNR & 77.5273 & 60.3738 & 49.2073 & 23.8819 & 31.0134\\ \hline \hline
\multirow{2}{*}{128 MB} & \% FPR & 7.5909 & 13.8357 & 16.4207 & 23.192 & 23.5977\\ \cline{2-7}
& \% FNR & 68.0553 & 49.2204 & 41.7178 & 22.89 & 22.2474\\ \hline \hline
\multirow{2}{*}{256 MB} & \% FPR & 5.0231 & 7.7876 & 8.3846 & 10.4808 & 11.3415\\ \cline{2-7}
& \% FNR & 54.0293 & 35.2653 & 31.2527 & 17.2991 & 11.5338\\ \hline \hline
\multirow{2}{*}{512 MB} & \% FPR & 2.7427 & 3.2263 & 3.33317 & 3.7953 & 4.1133\\ \cline{2-7}
& \% FNR & 32.8335 & 22.3998 & 20.2695 & 10.8729 & 4.2852\\ \hline \hline
\end{tabular}
\caption{Synthetic Dataset of 1B elements (90\% distinct)}
\label{1B:90distinct}
\end{center}
\end{table}

Hence, we observe that for both synthetic and real datasets of upto 1 billion records and varying percentage of distinct elements, the proposed algorithms in this work clearly outperforms SBF 
in terms of FNR, attaining an improvement of more than 300x times in certain cases. Also, for reasonable amount of memory FPR performance of all the algorithms turn out to be similar. Coupled 
with enhanced convergence rates compared to that of SBF, we present novel and efficient algorithms for the de-duplication problem.

\section{Conclusions and Future Work}
\label{sec:conc}

Real-time de-duplication or data redundancy removal for streaming datasets poses a challenging problem. In this work we have presented novel Bloom Filter based 
algorithms to tackle the problem efficiently. Using a novel combination of reservoir sampling and Bloom Filters we have proposed RSBF to obtain enhanced FNR 
and faster convergence to stability at comparable FPR with that of SBF. We further proposed BSBF encompassing a biased sampling method with Bloom Filters to 
obtain better FNR for varied applications requiring very low FNR tolerance. Variations of BSBF have also been presented in this work with different deletion 
designs to counter the effects to multiple element deletion in Bloom Filters. Finally a randomized load balanced algorithms has also been presented to provide 
a balanced performance on both the FPR and FNR fronts. These features make the proposed algorithms extremely efficient and applicable to real life scenarios. 

Using detailed theoretical results, we have proven the enhanced performance of the proposed algorithms in terms of FNR and convergence rates (stability). We 
demonstrate real-time in-memory DRR using both real and synthetically generated datasets of upto 1 billion records. We show FNR improvement over a vast 
range from 2x to 300x over existing results. To the best of our knowledge this work achieves the best FNR and convergence rates known with the same memory 
requirements as that of the competing algorithms. In future, we hope to study the effects of other biasing and sampling functions to further decrease the 
FNR. Investigations over the use of other structures and parallelizing the proposed algorithms may in turn lead to further enhancement and advancements 
in the field of parallel data redundancy removal research.





\bibliographystyle{elsarticle-num}
\bibliography{els}







\end{document}